\documentclass[smallextended]{svjour3}  
\newcommand{\argmax}[1]{\underset{#1}{\text{argmax}}}
\newcommand{\argmin}[1]{\underset{#1}{\text{argmin}}}
\usepackage{mathtools}
\DeclarePairedDelimiter{\ceil}{\lceil}{\rceil}
\DeclarePairedDelimiter{\floor}{\lfloor}{\rfloor}

\usepackage{graphicx}
\usepackage{algorithm,lipsum,changepage}
\usepackage[noend]{algpseudocode}
\usepackage{fullwidth}
\usepackage{enumitem}
\usepackage{pifont}
\usepackage{verbatim}
\usepackage{amssymb}

\usepackage[utf8]{inputenc}
\usepackage{cite}

\usepackage{fix-cm}
\begin{document}
	\title{Robust Monotone Submodular Function Maximization}
		\author{James B. Orlin \and Andreas S. Schulz \and Rajan Udwani}
		\institute{ \{jorlin,schulz,rudwani@mit.edu\}}
	
	\date{\vspace{-5ex}}	
	\maketitle

	\begin{abstract} 
		\vspace{-3mm}
		
		We consider a robust formulation, introduced by Krause et al. (2008), of the classical cardinality constrained monotone submodular function maximization problem, and give the first constant factor approximation results. The robustness considered is w.r.t.\ adversarial removal of up to $\tau$ elements from the chosen set. For the fundamental case of $\tau=1$, we give a deterministic $(1-1/e)-1/\Theta(m)$ approximation algorithm, where $m$ is an input parameter and number of queries scale as $O(n^{m+1})$. In the process, we develop a deterministic $(1-1/e)-1/\Theta(m)$ approximate greedy algorithm for bi-objective maximization of (two) monotone submodular functions. Generalizing the ideas and using a result from Chekuri et al. (2010), we show a randomized $(1-1/e)-\epsilon$ approximation for constant $\tau$ and $\epsilon\leq \frac{1}{\tilde{\Omega}(\tau)}$, making $O(n^{1/\epsilon^3})$ queries. Further, for $\tau\ll \sqrt{k}$, we give a fast and practical 0.387 algorithm. Finally, we also give a black box result result for the much more general setting of robust maximization subject to an Independence System.
		
	\end{abstract}
	
	\newtheorem{lem}{Lemma}
	\newtheorem{coro}[lem]{Corollary}
	\newtheorem{thm}[lem]{Theorem}
	\newtheorem{cnj}[lem]{Conjecture}
	\section{Introduction}
	A set function $f:2^N\rightarrow \mathbb{R}$ on the ground set $N$ 
	is called submodular if,
	$$f(A+a)-f(A)\leq f(B+a)-f(B) \text{  for all } B\subseteq A\subseteq N \text{ and } a\in N\setminus A .$$
	The function is monotone if $f(B)\leq f(A) \text{ for all } B\subseteq A$. We also impose $f(\emptyset)=0$, which combined with monotonicity implies non-negativity. Optimization problems with submodular objective functions have received a lot of interest due to several applications where instances of these problems arise naturally. However, unlike the (unconstrained) minimization of submodular functions, for which polytime algorithms exist \cite{min1,min2}, even the simplest maximization versions are NP-hard \cite{hard,hard1,hard2,feige}. In fact, they encompass several fundamental hard problems, such as max-cut, max-$k$-coverage, max-dicut and variations of max-SAT and max-welfare. 
	
	A long line of beautiful work has culminated in fast and tight approximation algorithms for many settings of the problem. As an example, for unconstrained maximization of non-monotone submodular functions, Feige et al.\ in \cite{feige}, provided an algorithm with approximation ratio of 0.4 and showed an inapproximability threshold of 1/2 in the value-oracle model.  Extensions by Gharan and Vondr\'ak \cite{gharan} and subsequently by Feldman et al.\ \cite{feld} led to further improvement of the guarantee (roughly 0.41 and 0.42, respectively). Finally, Buchbinder et al.\ in \cite{buch} gave a tight randomized $1/2$ approximation algorithm, and this was recently derandomized  \cite{deran}.
	
	Here we are interested in the problem of maximizing a monotone submodular function subject to a cardinality constraint, written as: 
	$P1:= \underset{A\subseteq N, |A|\leq k}{\max} f(A) .$
	The problem has been well studied and instances of $P1$ arise in several important applications, two of them being:
	
	\noindent \textbf{Sensor Placement} \cite{sense1,sense2,sense3,main}: Given a large number of locations $N$, we would like to place up to $k$ sensors at certain locations so as to maximize the \textit{coverage}. Many commonly used coverage functions measure the cumulative information gathered in some form, and are thus monotone (more sensors is better) and submodular (decreasing marginal benefit of adding a new sensor). 
	
	However, as highlighted in \cite{main}, it is important to ask what happens if some sensors were to fail. Will the remaining sensors have good coverage regardless of which sensors failed, or is a small crucial subset responsible for most of the coverage?  
	
	\noindent \textbf{Feature Selection} \cite{feature1,feature2,ml,main}: In many machine learning models, adding a new feature to an existing set of features always improves the modeling power (monotonicity) and the marginal benefit of adding a new feature decreases as we consider larger sets (submodularity).  Given a large set of features, we would like to select a small subset such that, we reduce the problem dimensionality while retaining most of the information. 
	
	However, as discussed in \cite{ml,main}, in situations where the nature of underlying data is uncertain, leading to non-stationary feature distributions, it is important to not have too much dependence on a few features. Taking a concrete example from \cite{ml}, in document classification, features may take not standard values due to small sample effects or in fact, the test and training data may come from different distributions. In other cases, a feature may even be deleted at some point, due to input sensors failures for example. Thus, similar questions 
	arise here too and we would like to have an `evenly spread' dependence on the set of chosen features. With such issues in mind, consider the following \textit{robust} variant of the problem, introduced in \cite{main},
	
	$$P2:= \underset{A\subseteq N, |A|\leq k}{\max}\quad \underset{Z\subseteq A, |Z|\leq \tau}{\min} \, f(A-Z) .$$
	
	Note that the parameter $\tau$ controls the degree of robustness of the chosen set since the larger $\tau$ is, the larger the size of subset $Z$ that can be adversarially removed from the chosen set $A$. For $\tau=0$, $P2$ reduces to the $P1$. Since this formulation optimizes the worst case scenario, a natural variation is to optimize the average case failure scenario \cite{krau}. However, this is not suitable for some applications. For instance, we may have no prior on the failure/deletion mechanism and furthermore, in critical applications, such as sensor placement for outbreak detection \cite{main,sense3}, we want protection against the worst case.
	This form of worst case analysis has been of great interest in operations research and beyond, under the umbrella of \textit{robust optimization} (e.g., \cite{book,bert,sim,bertsim}). The idea is to formulate the uncertainty in model parameters through a deterministic \textit{uncertainty set}. 
	While much work in this area assumes that the uncertainty set is a connected, if not convex set, the uncertainty set in $P2$, when $\tau=1$ for instance, is the disconnected set of canonical unit vectors $e_i\in \mathbb{R}^{N}$ (1 at entry $i$, 0 otherwise).
	
	\paragraph{Previous work on $P1$ and $P2$.} The first rigorous analysis of $P1$ was by Nemhauser et al.\ \cite{nem,nem1} in the late 70's, where they showed that the greedy algorithm gives a guarantee of $(1-1/e)$ and that this is best possible in the value-oracle model. Later, Feige \cite{hard} showed that this is also the best possible under standard complexity assumptions (through the special case of Max-$k$-cover). 
	On the algorithmic side, Badanidiyuru and Vondr\'ak \cite{fast} recently gave a faster algorithm for $P1$ that improved the quadratic query complexity of the classical greedy algorithm to nearly linear complexity, by trading off on the approximation guarantee. However, the optimality (w.r.t.\ approximation guarantee) of the greedy algorithm is rather coincidental, and for many complex settings of the problem (monotone or not), the greedy algorithm tends to be sub-optimal (there are exceptions, like \cite{denko}). An approach first explored by Calinescu et al.\ \cite{vond1}, that has been very effective, is to perform optimization on the \textit{multilinear extension} of the submodular function, followed by clever \textit{rounding} to get a solution to the original problem. 
	Based on this framework, tremendous progress has been made over the last decade for both monotone and non-monotone versions with various kinds of constraints \cite{vond1,welf,multiuni,hard1,vond,swap}. In fact, a general framework for establishing hardness of many of these variants \cite{hard1,hard2}, also relies intricately on properties of this relaxation.
	
	Moving on to $P2$, as we will see, the well known greedy algorithm and also the above mentioned \textit{continuous} greedy approach for $\tau=0$, can be arbitrarily bad even for $\tau=1$. In fact, many natural approaches do not have a constant factor guarantee for $\tau\geq 1$. The paper by Krause et al.\ \cite{main}, which formally introduced the problem, actually gives a bi-criterion approximation to the much more general and inapproximable problem:
	$\underset{A\subseteq N, |A|\leq k}{\max} \quad \underset{i\in\{1,2,\dots,m\}}{\min}\, f_i(A)$,
	where $f_i(.)$ is monotone submodular for every $i$. Their algorithm, which is based on the greedy algorithm for $P1$, when specialized to $P2$, guarantees optimality by allowing sets up to size $k(1+\Theta (\log (\tau k\log n)))$ and has runtime exponential in $\tau$. To the best of our knowledge, no stronger/constant factor approximation results were known for $P2$ prior to our work. 
	\vspace{-2mm}
	\paragraph{Our contributions:} We work in the value oracle model
	and give constant factor guarantees for $P2$ with combinatorial, `greedy like' algorithms. To ease presentation, we will usually ignore factors of the form $\big(1-\frac{O(1)}{k}\big)$ in the approximation guarantees and thus, most of the results presented here are asymptotic in $k$. Our initial focus is on a restricted case, where we construct relatively simple algorithms and get insights that generalize. For this special case, we give a $(1-1/e)$ algorithm for $\tau=o(k)$. 
	
	In the non-restricted setting, we first focus on generalizing results for $\tau=1$, 
	for which we propose a fast and practical 0.5547 approximation and later an asymptotically $(1-1/e)-1/\Theta(m)$ algorithm (with runtime exponential in $m$, which is an input parameter). This relies on developing a new $(1-1/e)-1/\Theta(m)$ approximate greedy algorithm for the bi-objective maximization of monotone submodular functions subject to cardinality constraint, of which the problem of maximizing the minimum of two arbitrary monotone submodular functions is a special case. We conjecture that the greedy algorithm can be generalized to work for multi-objective maximization of a fixed number of monotone submodular functions. There has been previous work on the multi-objective problem for a constant number of monotone submodular functions and a randomized $(1-1/e)-\epsilon$ algorithm for (a much more general version of) this problem was given in \cite{swap}. It uses the continuous greedy algorithm from \cite{vond}, along with an innovative dependent rounding scheme called \emph{swap rounding}, and we use this algorithm as a subroutine to get a randomized $(1-1/e)-\epsilon$ approximation for the robust problem when $\tau$ is constant, however the runtime scales as $O(n^{1/\epsilon^3})$. Finally, for $\tau=o\big(\sqrt{\frac{k}{c(k)}}\big)$, we give a fast $0.387\big(1-\frac{1}{\Theta(c(k))}\big)$ algorithm where $c(k)\xrightarrow{k\to \infty}\infty$ is an input parameter that governs the trade off between how large $\tau$ can be and how fast the guarantee converges to 0.387.
	
	In the more general case, where we wish to find a robust set $A$ in an independence system, we extend some of the ideas from the cardinality case into an enumerative procedure that yields an $\alpha/(\tau+1)$ approximation using an $\alpha$ approximation algorithm for $\tau=0$ as a subroutine. However, the runtime scales as $n^{\tau+1}$.  
	

	
	The outline for the rest of the paper is as follows: In Section 2, we introduce some notation and see how several natural ideas fail to give any approximation guarantees. In Section 3, we start with a special case and slowly build up to an algorithm with asymptotic guarantee $(1-1/e)-\epsilon$ for constant $\tau$, covering the other results in the process. Finally, in Section 4, we extend some of the ideas to more general constraints. Section 5 concludes with some open questions. 
	
	\section{Preliminaries}
	\subsection{Definitions} 
	We denote an instance of $P2$ on ground set $N$ with cardinality constraint parameter $k$ and robustness parameter $\tau$ by $(k,N,\tau)$. Subsequently, we use $OPT(k,N,\tau)$ to denote an optimal set for the instance $(k,N,\tau)$. For any given set $A$, we call a subset $Z_{\tau}$ a \textit{minimizer} if $f(A-Z_{\tau})= \underset{B\subseteq A; |B|= \tau}{\min}f(A-B)$. Also, let $\mathcal{Z}_{\tau}(A)$ be the set of minimizers of $A$. When $\tau=1$, we often use the letter $z$ for minimizers and when $\tau$ is otherwise clear from the context and fixed during the discussion we use the shorthand $Z,\mathcal{Z}(.)$. 
	Based on this we also introduce a key function $g_{\tau}(A)=f(A-Z_{\tau})$. Again, we simply use $g(.)$, when $\tau$ is clear from context. We generally refer to singleton sets without braces $\{\}$ and use $+,-$ and $\cup, \backslash$ interchangeably. Also, define the marginal increase in value due to a set $X$, when added to the set $A$ as $f(X|A)=f(A\cup X)-f(A)$. Similarly, $g(X|A)=g(A\cup X)-g(A)$.
	
	Let $\beta(\eta,\alpha)=\frac{e^\alpha-1}{e^\alpha-\eta}\in [0,1]$ for $\eta \in[0,1],\, \alpha\geq 0$. Note that $\beta(0,1)=(1-1/e)$. This function appears naturally in our analysis and will be useful for expressing approximation guarantees of the algorithms. Next, recall the widely popular greedy algorithm for $P1$:
	
	{\centering
		\vspace{-4mm}
		\begin{minipage}{\textwidth}
			\begin{algorithm}[H]
				\caption{Greedy Algorithm}
				\label{greedy}
				\begin{algorithmic}[1]
					\State Initialize $A=\emptyset$
					\While{$|A|< k$} $A=A+\argmax{x\in N-A}f(x|A)$ \EndWhile
					\vspace{-3mm}
					\State Output: $A$
				\end{algorithmic}
			\end{algorithm}
		\end{minipage}
	}
	
	Let $k=n$ in the above and denote the $i$-th element added by $a_i$. 
	Using this we index the elements in $N$ in the order they were added, so $N=\{a_1,a_2,\dots,a_n\}$.
	
	Also, recall the following theorem:
	\begin{lem}[Nemhauser, Wolsey \cite{nem,nem1}]\label{nem}
	 For all $\alpha\geq 0$, greedy algorithm terminated after $\alpha k$ steps yields a set $A$ with $f(A)\geq \beta(0,\alpha)f(OPT(k,N,0))$.
	\end{lem}
	For the sake of completeness, we include the proof in Appendix \ref{app1}.
	In addition, the following lemma, which compares the optimal value of $(k,N,\tau)$ with $(k-\tau,N,0)$ will be very useful:
	\begin{lem} \label{f and g} 
		\label{fng}
		For instances $(k,N,\tau)$, we have: $$g_{\tau}(OPT(k,N,\tau))\leq f(OPT(k-\tau,N-X,0))\leq f(OPT(k-\tau,N,0)),\quad $$ $\text{for all } X\subseteq N, |X|\leq \tau $.
	\end{lem} 
	\begin{proof}
		
		We focus on the first inequality since the second follows by definition.
		Let $x=\big|X\cap OPT(k,N,\tau)\big|\leq \tau$, then note that $g_{\tau}(OPT(k,N,\tau))\leq g_{\tau-x}(OPT(k,N,\tau)-X)$, since the RHS represents the value of some subset of $OPT(k,N,\tau)$ of size $k-\tau$, which upper bounds the LHS by definition. Now, $g_{\tau-x}(OPT(k,N,\tau)-X)\leq g_{\tau-x}(OPT(k-x,N-X,\tau-x))$ by definition. Finally, note that $g_{\tau-x}(OPT(k-x,N-X,\tau-x))\leq f(OPT(k-\tau,N-X,0))$ since the LHS represents the value of a set of size $k-x-(\tau-x)=k-\tau$ that does not include any element in $X$, giving us the desired.\qed
		%
	\end{proof}  
	Finally, it is natural to expect that we cannot approximate $P2$ better than $P1$ (which is approximable up to a factor of $\beta(0,1)$) and this is indeed true, as stated below. The proof is included in Appendix \ref{app1.1}. 
		\begin{lem}
		There exists no polytime algorithm with approximation ratio greater than $(1-1/e)$ for $P2$ unless $P=NP$. For the value oracle model, we have the same threshold, but for algorithms that make only a polynomial number of queries.
		\end{lem}

	\subsection{Negative Results}\label{ggreed}
	
	The example below demonstrates why the greedy algorithm that does well for instances of $P1$, fails for $P2$. However, the weakness will also indicate a property which will guide us towards better guarantees later.
	\vspace{-3mm}
	\paragraph{Example:} Consider a ground set $N$ of size $2k$ such that $f(a_1)=1$, $f(a_i)=0$, $\forall\, 2\leq i\leq k$ and $f(a_j)=\frac{1}{k}$, $\forall j\geq k+1$. Also, for all $j\geq k+1$, let $f(a_j|X)=\frac{1}{k}$ if $X\cap \{a_1,a_j\}=\emptyset$ and 0 otherwise. Consider the set $S=\{a_{k+1},\cdots,a_{2k}\}$ and let the set picked by the greedy algorithm (with arbitrary tie-breaking) be $A=\{a_1,\cdots,a_k\}$. Then we have that $f(A-a_1)=0$ and $f(S-a_j)=1-\frac{1}{k}$ for every $a_j\in S$. 
	The insight here is that greedy may select a set where only the first few elements contribute most of the value in the set, which makes it non-robust. However, as we discuss more formally later, such a concentration of value implies that only the first two elements, $\{a_1,a_2\}$, are critical and protecting against removal of those suffices for best possible guarantees. 
	
	In fact, many natural variations fail to give an approximation ratio better than $1/(k-1)$. Indeed, a guarantee of this order, i.e.\ $1/(k-\tau)$, is achievable for any $\tau$ by the following na\"{\i}ve algorithm: $\textit{Pick the $k$ largest value elements}$. 
	It is also important to examine if the function $g$ is super/sub-modular, since that would make existing techniques useful. It turns out not. 
	However, it is monotonic. Despite this, it is interesting to examine a natural variant of the greedy algorithm, where we greedily pick w.r.t $g$, but that variant can also be arbitrarily bad if we pick just one element at each iteration.  
	\section{Main Results}
	Before we start, remember that the focus of these results is on asymptotic performance guarantee (for large $k$). In some cases, the results can be improved for small $k$ but we generally ignore these details. 
	
	Additionally, in every algorithm that uses the greedy algorithm (Algorithm \ref{greedy}) as a subroutine, especially the fast and practical algorithms \ref{copyopt} ,\ref{a0.55} and \ref{greedytau}, we can replace the greedy addition rule of adding $x=\arg\max_{x\in S_1} f(x|S_2)$ for some $S_1,S_2$, by the more efficient thresholding rule in \cite{fast}, where, given a threshold $w$, we add a new element $x\in S_1$ if $f(x|S_2)\geq w$. This improves the query/run time to $O(\frac{n}{\epsilon}\log {\frac{n}{\epsilon}})$, at the cost of a factor of $(1-\epsilon)$ in the guarantee. 
	\subsection{Special case of ``copies''}\label{special}
	We first consider a special case, which will serve two purposes. First, it will simplify things and the insights gained for this case will generalize well. Secondly, since this case may arise in practical scenarios, it is worthwhile to discuss the special algorithms as they are much simpler than the general algorithms discussed later. 
	
	
	Given an element $x\in N$, we call another element $x'$ a \textit{copy} of element $x$ if, 
	$$f(x')=f(x)\text{ and } f(x'|x)=0.$$
	This implies $f(x|x')=f(\{x,x'\})-f(x')=f(x)+f(x'|x)-f(x')=0$. In fact, more generally, $f(x'|A)=f(x|A)$ for every $A\subseteq N$, since $f(A+\{x,x'\})=f(A+x)=f(A+x')$.
\emph{This is a useful case for robust sensor placement, if we were allowed to place/replicate multiple sensors at certain locations that are critical for coverage.} Assume that each element in $N$ has $\tau$ copies. In the next section we discuss algorithms for this special case when $\tau=1$. This will help build a foundation, even though the results therein are superseded by the result for $\tau=o(k)$.
	%
	%

	\subsubsection{Algorithms for $\tau=1$ in presence of ``copies''}\label{cpy1sec}
	Let $a'_i$ denote the copy of element $a_i$. As briefly indicated previously, we would like to make our set robust to removal of critical elements. In the presence of copies, adding a copy of these elements achieves that. So as a first step, consider a set that includes a copy of each element, and so is unaffected by single element removal. One way to do this is to run the greedy algorithm for $k/2$ iterations and then add a copy of each of the $k/2$ elements. Then, it follows from Lemma \ref{nem}
	that $g(S)=f(S)\geq \beta(0,\frac{k/2}{k-1})f(OPT(k-1,N,0))$ and then from Lemma \ref{fng} we have, $g(S)\geq (1-\frac{1}{\sqrt{e}})\, g(OPT(k,N,1))$, where we use the fact that $\beta (0,\frac{k/2}{k-1})>\beta(0,0.5)=(1-\frac{1}{\sqrt{e}})$. Hence, we have $\approx 0.393$-approximation and the bound is tight. We can certainly improve upon this, one way to do better is to think about whether we really need to copy all $k/2$ elements. It turns out that copying just $a_1$ and $a_2$ is enough! Intuitively, if the greedy set has value nicely spread out, we could get away without copying anything but nevertheless, in such a case copying just two elements does not affect the value much. Otherwise, as in the example from Section 2, if greedy concentrates its value on the first few elements, then copying them is enough. 
	Before we state and prove this formally, consider the below corollary:
	\begin{coro}\label{init}
		Let $A$ be the final set obtained by running the greedy algorithm for $l$ steps on an initial set $S$. Then we have,
		$$ f(A-S|S)\geq \beta\Big(0,\frac{l}{k}\Big)f(OPT(k,N,0)|S) \geq \beta\Big(0,\frac{l}{k}\Big)\Big(f\big(OPT(k,N,0)\big)-f(S)\Big)$$
	\end{coro}
	\begin{proof}
		Follows from Lemma \ref{nem}, since $f(.|S)$ is monotone submodular for any fixed set $S$ and union of greedy on $f(.|S)$ with $S$ is the same as doing greedy on $f(.)$ starting with $S$.
	\qed\end{proof}
	Suppose $A$, referred to as the greedy set, is the output of Algorithm \ref{greedy}. We now state and prove a simple yet key lemma, which will allow us to quantify how the ratio $\frac{f(A)}{f(OPT(k,N,0))}$ improves over $(1-1/e)$, as a function of how concentrated the value of the greedy set is on the first few elements.
	\begin{lem}\label{key}
		Starting with initial set $S$, run $l$ iterations of the greedy algorithm and let $A$ be the output (so $|A|=l+|S|$). Then given $f(S)\geq c f(A)$ for some $c\leq 1$ (high concentration of value on $S$ implies large $c$), we have, 
		$$f(A)\geq \beta\Big(c,\frac{l}{k}\Big) f(OPT(k,N,0)), \text{  for all }k\geq 1$$
		In a typical application of this lemma, we will have $S$ be the first $s=|S|$ elements of the greedy algorithm on $\emptyset$ and $c=s\eta$ for some $\eta\leq 1/s$. Additionally, $|A|$ can be greater than $k$.
	\end{lem}
	\begin{proof}	
		Denote $f(OPT(k,N,0))$ by $OPT$ and with the sets $A,S$ as defined above, let $\delta=\frac{f(A)}{OPT}$, which implies $f(S)\geq c\delta OPT$ (by assumption). Also, since we allow $|A|$ to be larger than $k$, assume $\delta\leq 1$ (otherwise statement is true by default).
		
		Now from Corollary \ref{init} we have, $	f(A)= f(S) + f(A-S|S) \geq f(S) +\beta(0,\frac{l}{k}) (OPT-f(S))$ and thus,
		\begin{eqnarray*}
			&\delta OPT  \geq  (1-\beta(0,\frac{l}{k}))\times c \delta f(OPT)+ \beta(0,\frac{l}{k}) f(OPT) \quad  [\text{substitution}] &\notag\\
			&\delta (1-c(1-\beta(0,\frac{l}{k}))) \geq \beta(0,\frac{l}{k})&\\
			&\delta (1-c\frac{1}{e^{l/k}}) \geq \frac{e^{l/k}-1}{e^{l/k}}&\\
			&\delta \geq \frac{e^{l/k}-1}{ e^{l/k} -c}=\beta\Big(c,\frac{l}{k}\Big) &\\
			&\implies f(A) \geq \beta\Big(c,\frac{l}{k}\Big)f(OPT(k,N,0)).\quad \qed 
		\end{eqnarray*}
	\end{proof}  
	Let us understand the above lemma with a quick example. Consider the greedy set of the first $k$ elements $A=\{a_1,\dots,a_k\}$ and let $f(OPT(k,N,0))=1$. Now, clearly $f(a_1)\geq \frac{f(A)}{k}$. So using the lemma with $S=\{a_1\}$ and $c=1/k$ gives us simply that $f(A)\geq \beta\Big(\frac{1}{k},\frac{k-1}{k}\Big)\approx (1-1/e)$, as expected. Next, if $f(a_1)\geq f(A)/2$, then we have that $f(A)\geq \beta\Big(\frac{1}{2},\frac{k-1}{k}\Big)\approx 0.77$, asymptotically much better than $(1-1/e)\approx 0.63$. Similarly, if $f(\{a_1,a_2\})\geq f(A)/2$ we again have $f(A)\geq \beta\Big(\frac{1}{2},\frac{k-2}{k}\Big)\approx 0.77$. Additionally, we could compare the value $f(A)$ to $f(OPT(k-1,N,0))$ instead, and then replacing $k$ by $k-1$ in the denominator, we have $f(A)\geq \beta\Big(\frac{1}{2},\frac{k-2}{k-1}\Big)f(OPT(k-1,N,0))$. 

	\emph{Now, consider the algorithm that copies the first two elements and thus outputs: $\{a_1,a'_1,a_2,a'_2,a_3,\dots,a_{k-2}\}$. Call this the 2-Copy algorithm.} Using the above lemma, we show that this algorithm gives us the best possible guarantee asymptotically, aligning with the intuitive argument we presented earlier. Preceding the actual analysis, we show an elementary lemma that we will use frequently.
		\begin{lem}\label{subkey}
		For $0\leq \eta\leq \frac{1}{3}$ and arbitrary $\alpha$, $\argmin{\eta} (1-\eta)\beta(3\eta,\alpha)=0.$
		\end{lem}
		\begin{proof}
				\begin{eqnarray*}
					(1-\eta)\beta(3\eta,\alpha)&& =\, \frac{1}{3}(3-3\eta) \frac{e^{\alpha}-1}{e^{\alpha}-3\eta}\, =\frac{1}{3}(e^{\alpha}-1)\Big(1+\frac{3-e^{\alpha}}{e^{\alpha}-3\eta}\Big)\\
					&&\geq \frac{1}{3}(e^{\alpha}-1)\Big(1+\frac{3-e^{\alpha}}{e^{\alpha}}\Big)\, = (1-1/e^{\alpha})\, = \beta(0,\alpha)
				\end{eqnarray*}
			\end{proof}
	\begin{thm}\label{2cpy}
		For the case with copies,  2-Copy is $\beta\big(0,\frac{k-5}{k-1}\big)\xrightarrow{k\to \infty}(1-1/e)$ approximate.
	\end{thm}
	\begin{proof}
		First, denote the output as $A=\{a_1,a'_1,a_2,a'_2,a_3,\dots,a_{k-2}\}$ and notice that $f(A)=f(\{a_1,a_2,\dots,a_{k-2}\})$. As a warm-up, using Lemma \ref{nem} we get, 
		\begin{eqnarray}
		f(A)\geq \beta\Big(0,\frac{k-2}{k-1}\Big) f(OPT(k-1,N,0))\geq \beta\Big(0,\frac{k-2}{k-1}\Big)g(OPT(k,N,1)),\label{bsc} 
		\end{eqnarray}
		 where the second inequality follows from Lemma \ref{fng}.
		
		Let $z$ be a minimizer of $A$. We can assume that $z\not\in \{a_1,a'_1,a_2,a'_2\}$ since all of these have 0 marginal. Now let $f(z|A-z)=\eta f(A)$. 
		We have due to greedy additions and submodularity, $f(a_3|\{a_1,a_2\})\geq f(z|\{a_1,a_2\})\geq  f(z|A-z)$ and similarly, $f(\{a_1,a_2\})\geq 2f(z|A-z)$. This implies that $f(\{a_1,a_2,a_3\})\geq 3\eta f(A)$, which relates the value removed by a minimizer $z$ to the value concentrated on the first 3 elements $\{a_1,a_2,a_3\}$. The higher the value removed, the higher the concentration and closer the value of $f(A)$ to $f(OPT(k-1,N,0))$. More formally, with $S=\{a_1,a'_1,a_2,a'_2,a_3\}$, $k$ replaced by $k-1$, $l=k-5$ and $c=3\eta$ ($\eta\leq 1/3$), we have from Lemma \ref{key},
		$$f(A)\geq  \beta\Big(3\eta,\frac{k-5}{k-1}\Big) f(OPT(k-1,N,0))\geq \beta\Big(3\eta,\frac{k-5}{k-1}\Big) g(OPT(k,N,1)).$$
		Which implies that $g(A)=(1-\eta)f(A)\geq (1-\eta)\beta(3\eta,\frac{k-5}{k-1}) g(OPT(k,N,1)).$ Applying Lemma \ref{subkey} with $\alpha=\frac{k-5}{k-1}$ finishes the proof.
		As an example, for $k\geq55$ the value of the ratio is $\geq 0.6$. Additionally, we can use more precise bounds of the greedy algorithm for small $k$ to get better guarantees in that regime. \qed 
	 \end{proof}
	The result also implies that for large $k$, if the output set $A$, of Algorithm \ref{greedy} (the greedy algorithm), has a minimizer $a_i$ with $i\geq 3$, then $g(A)=f(A-a_i)\geq (1-1/e)g(OPT(k,N,1))$, i.e. the greedy algorithm is $(1-1/e)$ approximate for the robust problem in this situation. This is because, for such an instance we can assume that the set contains copies of $a_1,a_2$ without changing anything and then Theorem \ref{2cpy} applies. 
	
	Moreover, if instead of just the first two, we copy the first $i$ elements for $i\geq 3$, we get the same guarantee but with worse asymptotics, so copying more than first two does not result in a gain. On the other hand, copying just one element, $a_1$, gives a tight guarantee of 0.5 (proof omitted). 
	
	Since $(1-1/e)$ is the best possible guarantee achievable asymptotically, we now shift focus to case of $\tau>1$ but $\ll k$, and generalize the above ideas to get an asymptotically $(1-1/e)$ approximate algorithm in presence of copies. 
	
		\subsubsection{$(1-1/e)$ Algorithms for $\tau=o(k)$ in the presence of ``copies''}\label{copygen}
		Assume we have $\tau$ copies available for each $a_i\in N$. As we did for $\tau=1$, we would like to determine a small critical set of elements, copy them (possibly many times) and then add the rest of the elements greedily to get a set of size $k$. In order to understand how large the critical set should be, recall that in the proof of Theorem \ref{2cpy}, we relied on the fact that $f(\{a_1,a_2\})$ is at least twice as much as the value removed by the minimizer, and then we could use Lemma \ref{key} to get the desired ratio. To get a similar concentration result on the first few elements for larger $\tau$, we start with an initial set of size $2\tau$ and in particular, we can start with $A_{2\tau}=\{a_1,a_2,\dots,a_{2\tau}\}$. Additionally, similar to the 2-Copy algorithm, we also want the set to be unaffected by removal of up to $\tau$ elements from $A_{2\tau}$. We do this by adding $\tau$ copies of each element in $A_{2\tau}$. More concretely, consider the algorithm that greedily picks the set $A_{k-2\tau^2}=\{a_1,a_2,\dots,a_{k-2\tau^2}\}$, and copies each element in $A_{2\tau}=\{a_1,a_2,\dots,a_{2\tau}\}\subseteq A_{k-2\tau^2}$, $\tau$ times. 
		 Denote the set of copies by $C(A_{2\tau})$ and we have $|C(A_{2\tau})|=2\tau^2$. To summarize, the algorithm outputs the set $$A=A_{2\tau}\cup C(A_{2\tau}) \cup \{a_{2\tau+1},\dots,a_{k-2\tau^2}\}.$$ Observe that when $\tau=1$, this coincides with the 2-Copy algorithm.
		
		We next show that this algorithm is $\beta(0,\frac{k-2\tau^2-3\tau}{k-\tau})\xrightarrow{k\to \infty} (1-1/e)$ approximate.
		\begin{proof}
			We can assume that $Z\cap (A_{2\tau}\cup C(A_{2\tau}))=\emptyset$ or alternatively $Z\subseteq A_{k-2\tau^2}-A_{2\tau}$, since there are $\tau+1$ copies of every element (counting the element itself) and $Z$ cannot remove all. Recall that in the analysis of the 2-Copy algorithm, we showed $f(\{a_1,a_2,a_3\})\geq 3f(z|A-z)$ and then used Lemmas \ref{key} and \ref{subkey} to get the desired. Analogously, here we would like to show and $f(A_{3\tau})\geq 3f(Z|A-Z)$ and do the same.
			
			Next, index elements in $Z$ from $1$ to $\tau$, with the mapping $\pi:\{1,\dots,\tau\}\rightarrow \{2\tau+1,\dots,k-2\tau^2\}$, such that $\pi(i)>\pi(j)$ for $i>j$ and $a_{\pi(i)}$ is the $i$-th element in $Z$. Then with $A_i=\{a_1,\dots,a_i\}$ for all $i$, we have by submodularity, the following set of inequalities,
			\begin{eqnarray}
			f(a_{\pi(i)}|A_{\pi(i)-1})&&= f(a_{\pi(i)}|A-\{a_{\pi(i)},\dots,a_{k-2\tau^2}\})\notag\\
			&&\geq f(a_{\pi(i)}|A-\{a_{\pi(i)},a_{\pi(i+1)},\dots,a_{\pi(\tau)}\})\quad \forall i\notag\\
			\implies \sum_{i=1}^{\tau} f(a_{\pi(i)}|A_{\pi(i)-1})&&\geq f(Z|A-Z)\label{sum}
			\end{eqnarray}
			where the RHS in (\ref{sum}) is by definition.
			
			Note that $\pi(i)>2\tau$, and for arbitrary $i\in\{1,\dots,\tau\}=[\tau]$, $j\in\{1,\dots,2\tau\}=[2\tau]$, due to greedy iterations we have that $f(a_{\pi(i)}|A_{\pi(i)-1})\leq f(a_j|A_{j-1})$. So consider any injective mapping from $i\in [\tau]$ to 2 distinct elements $i_1,i_2\in [2\tau]$, for instance $i_1=i,i_2=i+\tau$. We rewrite the previous inequality as, 
			$$2f(a_{\pi(i)}|A_{\pi(i)-1})\leq f(a_{i_1}|A_{i_1-1})+ f(a_{i_2}|A_{i_2-1})$$
			Summing over all $i$, along with (\ref{sum}) above gives,
			\begin{eqnarray}2f(Z|A-Z)\leq f(A_{2\tau})\label{iter}\end{eqnarray}
			where the RHS is by injective nature of mapping and definition of $f(.|.)$. 
			
			In fact, we have $\pi(i)\geq 2\tau+i$ and thus, $f(a_{\pi(i)}|A_{\pi(i)-1})\leq f(a_{2\tau+i}|A_{2\tau+i-1})$. 	From this, we have for the set $A_{3\tau}-A_{2\tau}=\{a_{2\tau+1},\dots,a_{3\tau}\}$ that, 
			\begin{eqnarray}f(Z|A-Z)\leq f(A_{3\tau}-A_{2\tau}|A_{2\tau})\label{iter2}\end{eqnarray}
			(\ref{iter}) and (\ref{iter2}) combined give us,
			$$f(A_{3\tau})\geq 3 f(Z|A-Z)$$
			We have from Lemma \ref{fng} that $f(OPT(k-\tau,N,0))\geq g(OPT(k,N,\tau))$ and combined with using Lemma \ref{key}, with $k$ replaced by $k-\tau$, $S=A_{3\tau} \cup C(A_{2\tau})$, $s=3$ and $l=k-|S|=k-2\tau^2-3\tau$ and $f(Z|A-Z)=\eta f(A)$ gives us,
			$$	f(A)\geq \beta\Big(3\eta,\frac{k-2\tau^2-3\tau}{k-\tau}\Big)g(OPT(k,N,\tau)).$$
			Thus $g(A)=(1-\eta)f(A)\geq \beta\Big(0,\frac{k-2\tau^2-3\tau}{k-\tau}\Big)g(OPT(k,N,\tau))$, follows using Lemma \ref{subkey}. \qed
		\end{proof}
		
		Note that while we could ignore the asymptotic factors in approximation guarantee for $\tau=1$, here we cannot, and this is where the upper bound on $\tau$ comes in. Recall that we compare the value $g(OPT(k,N,\tau))$, which is the $f(.)$ value of a set of size $k-\tau$, to the $f(.)$ value of a set of size $k-\Theta(\tau^2)$ (since the $\Theta(\tau^2)$ elements added as copies do not contribute any real value). Now $\frac{k-\Theta(\tau^2)}{k}$ converges to 1 only for $\tau\ll\sqrt{k}$ and it is this degradation that creates the threshold of $o(\sqrt{k})$. 
		
		However, it turns out that we don't need to add $\tau$ copies of each element in $A_{2\tau}$. Intuitively, the first few elements in $A_{2\tau}$ are more important and for those we should add $\tau$ copies, but the later elements are not as important and we can add fewer copies. 
		In fact, we can geometrically decrease the number of copies we add from $\tau$, to $1$, over the course of the $2\tau$ elements, resulting in a total of $\Theta(\tau\log\tau)$ copies. The resulting approximation ratio converges to $(1-1/e)$ for $\tau=o(k)$.
		
		More concretely, consider the following algorithm,
		
		{\centering
			\vspace{-4mm}
			\begin{minipage}{\textwidth}
				\begin{algorithm}[H]
					\caption{$(1-1/e)$ Algorithm for copies when $\tau=o(k)$}
					\label{copyopt}
					\begin{algorithmic}[1]
						\State Initialize $A=A_{2\tau}, i=1$
						\While{$i\leq\ceil{\log 2\tau}$} 
						\State $A=A\cup\big(\ceil{\tau/2^{i-1}}\text{ copies for each of } \{a_{2^i-1},\dots,a_{2^{i+1}-2}\}\cap A_{2\tau}\big)$; $i=i+1$
						\EndWhile
						\While{$|A|< k$} $A=A+\argmax{x\in N-A}f(x|A)$
						\EndWhile
						\vspace{-3mm}
						\State Output: $A$
					\end{algorithmic}
				\end{algorithm}
			\end{minipage}
			\vspace{-1mm}
			\centering}
		
		Essentially, we start with the set $A_{2\tau}$, add $\tau$ copies each for $\{a_1,a_2\}$, $\tau/2$ copies for each of $\{a_3,\dots,a_6\}$, $\tau/4$ copies for each of $\{a_7,\dots,a_{14}\}$ and so on, finally adding the rest of the $k-\Theta(\tau\log\tau)$ elements greedily. Notice that we could get the best possible guarantee with a minimizer oblivious algorithm i.e., the output is independent of the minimizer at any stage of the algorithm.
		\vspace{-3mm}
		\begin{thm}
			Algorithm \ref{copyopt} is $\beta\Big(0,\frac{k-\Theta(\tau\log\tau)}{k-\tau}\Big)\xrightarrow{k\to\infty}(1-1/e)$ approximate for $\tau=o(k)$.
		\end{thm} 
		\vspace{-3mm}
		\begin{proof}
			The basic outline of the analysis is similar to that of the previous one for $\tau=o(\sqrt{k})$, so we focus only on the differences here. 
			
			Let $A_1$ denote the subset of $A$ obtained in the final greedy phase (step 4 in Algorithm \ref{copyopt}). Note that if $Z\cap A_{2\tau}=\emptyset$, then we have (\ref{iter}) and (\ref{iter2}) as before and thus $f(A_{3\tau}) \geq 3f(Z|A-Z)$. By applying Lemma \ref{key} with $l=k-\Theta(\tau\log\tau)$ this time, we get the desired. However, unlike the previous analysis, here $Z\cap A_{2\tau}$ need not be empty since we have less than $\tau$ copies for many elements. We would like to show that $f(A_{2\tau}-Z)\geq 2 f(Z|A-Z)$ regardless. Let $m=\ceil{\log 2\tau}-1$ and in fact, for ease of presentation we assume that $\tau=2^{m}$. Also let $B_i=\{a_{2^i-1},\dots,a_{2^{i+1}-2}\}\cap A_{2\tau}$ i.e., the elements for which we add $\tau/2^{i-1}$ copies in the algorithm and let $C_i$ denote the set of copies of these elements. Note that $|B_i|=2^i$ and $|C_i|= 2\tau\geq 2|B_i|$ for all $i\leq m$, for $i=m+1$, $B_{m+1}=\{a_{2\tau-1},a_{2\tau}\}$ and $|C_{m+1}|=2$. We can assume that for every element in $A_{2\tau}$ included in minimizer $Z$, all copies of the element are also present in $Z$, hence $Z$ removes at most $\floor{\frac{\tau}{1+\tau/2^{i-1}}}=\floor{\frac{|B_i|}{2}\frac{\tau}{\tau+2^{i-1}}}\leq |B_i|/2-1$ elements from $B_i,\forall i\leq m$. So we assume w.l.o.g. $|Z\cap B_1|=0$.
			
			Let us first examine the case where $Z\subset B_i\cup C_i$ for some $2\leq i\leq m$ (the case of $i=m+1$ follows rather easily). Notice that $\sum_{j=1}^{i-1} |B_j|= 2(|B_i|/2-1)$, then from the observation above we have $2|Z\cap B_i|\leq 2\sum_{j=1}^{i-1} |B_j|$. This implies that we can injectively map any $|B_i|/2-1= 2^{i-1}-1$ elements in $B_i$, to two distinct elements with lower indices in $\cup_{j=1}^{i-1} B_j$. Then, just as we showed in (\ref{iter2}), we have an injective mapping from every element in $Z\cap A_{2\tau}$ to two distinct elements in $\cup_{j=1}^{i-1} B_j$, and this gives us $f(\cup_{j=1}^{i-1} B_j)\geq 2f(Z|A-Z)$ which further implies $f(A_{2\tau})\geq f(\cup_{j=1}^{i-1} B_j) +f(Z|\cup_{j=1}^{i-1} B_j)\geq 3f(Z|A-Z)$ and applying Lemmas \ref{key} and \ref{subkey} completes the case. 
			
			For the general case, let $x_j=|Z\cap B_j|, \forall j\in\{2,\dots,m+1\}$ (with $x_{m+1}\leq 2$) and let $x_{m+2}=|Z\cap A_1|$. From $|Z|=\tau$ and the fact that $Z$ contains all copies of every element in $|Z\cap A_{2\tau}|$, we have $\sum_{j=2}^{m+1} x_j(1+\frac{\tau}{2^{j-1}}) +x_{m+2}\leq \tau$. Observe that to show the existence of the desired injective mapping, it suffices to show,
			 \begin{eqnarray}
			 2x_i&&\leq |\cup_{j=1}^{i-1} B_j-Z|-\sum_{j=1}^{i-1}2x_j \notag\\
			 &&=|\cup_{j=1}^{i-1} B_j|-\sum_{j=1}^{i-1}3x_j, \forall i\in\{2,\dots,m+2\} \label{summary}
			 \end{eqnarray}
			
			Consider the polytope given by $X=(x_2,\dots,x_{m+2})$ and constraints $0\leq x_j\leq \floor{\frac{|B_i|}{2}\frac{\tau}{\tau+2^{i-1}}}\, \forall j$ and $\sum_{j=1}^{m+1} x_j(1+\frac{\tau}{2^{j-1}}) +x_{m+2}\leq \tau$. Then the extreme points correspond exactly to the special cases we considered so far i.e. (i)$Z\cap A_{2\tau}=\emptyset$ and (ii) $Z\subset B_i\cap C_i$, and we showed that the conditions (\ref{summary}) are satisfied for these cases. This implies the (linear) conditions (\ref{summary}) are satisfied for every point in the polytope and that completes the proof. \qed
		\end{proof}
	
	\subsection{Algorithms in the possible absence of ``copies''}
	To recap, thus far we chose a set of elements greedily and treated a suitably large subset of the initial few elements as `critical' and added `enough' copies of these elements to ensure that we keep a copy of each critical element in the set, even after adversarial removal. We aim to follow a similar scheme even in the absence of copies. In the general case however, we need to figure out a new way to ensure that our set is robust to removal of the first few critical elements chosen greedily. We first discuss how to approach this for $\tau=1$. 	
	\subsubsection{Algorithms for $\tau=1$ in the possible absence of ``copies''} \label{gensec}

	We start by discussing how one could construct a greedy set that is robust to the removal of $a_1$. In the case of copies, we would simply add a copy $a'_1$ of $a_1$, to accomplish this. Here, one approach would be to pick $a_1$ and then pick the rest of the elements greedily while ignoring $a_1$. Note that this would add a copy of $a_1$, if it were available, and if not it will permit the selection of elements which have small marginal on any set containing $a_1$, but possibly large marginal value in the absence of $a_1$. Such an element need not be selected by the standard Algorithm \ref{greedy} that doesn't ignore $a_1$. Formally,
	
	{\centering
		\vspace{-4mm}
		\begin{minipage}{\textwidth}
			\begin{algorithm}[H]
				\caption{0.387 Algorithm}
				\label{a0.387}
				\begin{algorithmic}[1]
					\State Initialize $A=\{a_1\}$
					\While{$|A|< k$} $A=A+\argmax{x\in N-A}f(x|A-\{a_1\})$
					\EndWhile
					\vspace{-3mm}
					\State Output: $A$
				\end{algorithmic}
			\end{algorithm}
		\end{minipage}
		\centering}
	
	This simple algorithm is in fact, asymptotically 0.387 approximate and the bound is tight 
	(proof in Appendix \ref{app2}). However, a possible issue with the algorithm is that it is oblivious to the minimizers of the set at any iteration. It ignores $a_1$ throughout, even if $a_1$ stops being the minimizer after a few iterations. Thus, if after every iteration, we check the minimizer and stop ignoring $a_1$ once it is not a minimizer (i.e. proceed with standard greedy iterations after such a point), we achieve a performance guarantee of 0.5 (proof omitted). Note that this matches the guarantee obtained by copying just $a_1$, in presence of copies, and so in this sense, we can build a set that is robust to removal of $a_1$ in the general setting. 
	
	As we saw for the case of copies, in order to get even better guarantees, we need to look at the set of the first two elements, $\{a_1,a_2\}$. A direct generalization of line 2 in Algorithm \ref{a0.387}, to a rule that ignores both $a_1$ and $a_2$, i.e.\  $\argmax{x\in N-A} f(x|A-\{a_1,a_2\})$, can be shown to have a performance bound less than $0.5$. In fact, many natural rules that were tried resulted in upper bounds $\leq 0.5$. Algorithm \ref{a0.55} avoids looking at both elements simultaneously and instead ignores $a_1$ until its marginal becomes sufficiently small and then does the same for $a_2$, if required. 
	
	{\centering
		\vspace{-3mm}
		\begin{minipage}{\textwidth}
			\begin{algorithm}[H]
				\caption{0.5547$-\Omega(1/k)$ Algorithm}
				\label{a0.55}
				\begin{algorithmic}[1]
					\State Initialize $A=\{a_1,a_2\}$
					
					\noindent \textbf{Phase 1:}
					\While{$|A|< k \textbf{ and } f(a_1|A-a_1)> \frac{f(A)}{3}$} $A=A+\argmax{x\in N-A}f(x|A-a_1)$
					\EndWhile
					\vspace{-3mm}
					\noindent \textbf{Phase 2:}
					\vspace{3mm}
					\While{$|A|< k \textbf{ and } f(a_2|A-a_2)> \frac{f(A)}{3}$} $A=A+\argmax{x\in N-A}f(x|A-a_2)$
					\EndWhile
					\vspace{-3mm}
					\noindent \textbf{Phase 3:}
					\vspace{3.5mm}
					\While{$|A|< k$} $A=A+\argmax{x\in N-A}f(x|A)$
					\EndWhile
					\vspace{-2.5mm}
					\State Output:$ A$
				\end{algorithmic}
			\end{algorithm}
		\end{minipage}
		\centering}
	
	The algorithm is asymptotically 0.5547-approximate (as an example, guarantee $>0.5$ for $k\geq50$) and note that it's minimizer oblivious and only uses greedy as subroutine, which makes it fast and easy to implement (recall that greedy can be replaced by thresholding). The analysis is presented in Appendix \ref{app3}. 

	Next, recall that we want to make the final set robust to removal of either one of the two elements $\{a_1,a_2\}$. 
	In order to improve upon the guarantee of Algorithm \ref{a0.55}, which deals with the two elements one at a time, first ignoring $a_1$ and then $a_2$ if required, we devise a way to add new elements while paying attention to both $a_1$ and $a_2$ simultaneously. 
	To this end, consider the following addition rule, 
	\begin{eqnarray*}
		&\argmax{|S|\leq m;S\subseteq N-A} g(S|A)=\argmax{|S|\leq m;S\subseteq N-A} \big[&\min\big\{f(S+A-a_1),f(S+A-a_2)\big\}-\\
		&&\min\big\{f(A-a_1),f(A-a_2)\big\}\big],
	\end{eqnarray*}
	 i.e. greedily adding $m$ tuples but w.r.t. to the $g(.)$ function now instead of $f(.)$, while $z\in\{a_1,a_2\}$, for suitable $m\geq 1$.
	We need to resort to $m$-tuples instead of singletons because, for $m=1$ we cannot guarantee improvements at each iteration, as there need not be any element that adds marginal value on both $a_1$ and $a_2$. However, for larger $m$ we can show improving guarantees.
	More concretely, consider an instance where $f(a_1)=f(a_2)=1$, $a_1$ has a copy $a'_1$ and additionally both $a_1$ and $a_2$ have `partial' copies, $f(a^j_i)=\frac{1}{k}$ and $f(a^j_i|a_i)=0$ for $j\in \{1,\dots,k\}$, $i\in\{1,2\}$. Also, let there be a set $G$ of $k-2$ garbage elements with $f(G)=0$. Finally, let $f(\{a_1,a_2\})=2$ and $f(a^j_i|X)=\frac{1}{k}$ if $a_i\not\in X$ (and also $a'_1\not\in X$ for $i=1$). Running the algorithm with: (i) $m=1$ outputs $\{a_1,a_2\}\cup G$ in the worst case, with (ii) $m=2$ outputs $a^j_1,a^j_2$ on step $j$ of Phase 1 and thus, `partially' copies both $a_1$ and $a_2$. Instead, if we run the algorithm with (iii) $m=3$, the algorithm picks up $a'_1$ and then copies $a_2$ almost completely with $\{a^1_2,\dots,a^{k-3}_2\}$. In fact, we will show that while $\{a_1,a_2\}$ are minimizers, adding $m$-tuples in this manner allows us to guarantee that at each step we increase $g(A)$ by $\frac{m-1}{m}\frac{1}{k}$ times the difference from optimal. Thus, when $m$ is large enough that $\frac{m-1}{m}\approx1$, we effectively add value at the `greedy' rate of $\frac{1}{k}$ times the difference from optimal (ref. Lemma \ref{nem}). However, this is while $z\in\{a_1,a_2\}$, so we need to address the case when $\{a_1,a_2\}$ are not minimizers. One approach would be to follow along the lines of Algorithm \ref{a0.55} by adding singletons greedily w.r.t. $f$ (similar to Phase 3) once the minimizer falls out of $\{a_1,a_2\}$. This is what we do in the algorithm below, recall that $\mathcal{Z}(A)$ is the set of minimizers of set $A$.
	
	{\centering
			\vspace{-3mm}
		\begin{minipage}{\textwidth}
			\begin{algorithm}[H]
				\caption{A $(1-1/e)-1/\Theta(m)$ Algorithm for $\tau=1$}
				\label{a0.63}
				\textbf{input: } $m$
				\begin{algorithmic}[1]
					\State Initialize $A=\{a_1,a_2\}$
					
					\noindent \textbf{Phase 1:} \While{$|A|< k \textbf{ and } \mathcal{Z}(A)\subseteq\{a_1,a_2\}$}
					\State $l=\min\{m,k-|A|\}$ 
					\State $A=A\cup \underset{|S|= l;S\subseteq N-A}{\text{argmax}} g(S|A)$
					\EndWhile
					\vspace{-3mm}
					\noindent \textbf{Phase 2:}
					\vspace{3mm}
					\While{$|A|< k$} $A=A+\argmax{x\in N-A}f(x|A)$
					\EndWhile
					\vspace{-3mm}
					\State Output:$ A$
				\end{algorithmic}
			\end{algorithm}
		\end{minipage}
		\vspace{0mm}
		\centering}

	Before analyzing the approximation guarantee of Algorithm \ref{a0.63}, we first need to show that in Phase 1, at each step we greedily add an $m$-tuple with marginal value $\frac{m-1}{k}$ times the difference from optimal. We do this by showing a more general property below,
	
	\begin{lem}\label{pareto}
		Given two monotone submodular functions, $f_1,f_2$ on ground set $N$. If there exists a set $S$ of size $k$, such that $f_i(S)\geq V_i, \forall i$, then for every $m$ ($2\leq m\leq k$), there exists a set $X\subseteq S$ with size $m$, such that, $$f_i(X)\geq \frac{m-1}{k}V_i,\forall i.$$
	\end{lem}
	\begin{proof}
		First, we show that it suffices to prove this statement for two modular functions $h_1,h_2$. Note that we can reduce our ground set to $S$. Now, consider an arbitrary indexing of elements in $S=\{s_1,\dots,s_k\}$ and let $S_j=\{s_1,\dots,s_j\}, \forall j\in [k]$. Consider modular functions such that value of element $s_j$ is $h_i(s_j)\vcentcolon= f_i(s_j|S_{j-1})$. Note that $h_i(S)=f_i(S)$ and additionally, by submodularity, we have that $h_i(X)$ is a lower bound on $f_i(X)$ i.e. for every set $X\subseteq S$, $f_i(X)=\sum_{j:s_j\in X}f_i(s_j|X\cap S_{j-1})\geq h_i(X)$. Also, we can assume w.l.o.g.\ that $h_i(S)=1,\forall i$ and so it suffices to show that $h_i(X)\geq \frac{m-1}{k}, \forall i$.
		
		We proceed by induction on $m$. For the base case of $m=2$, we can pick elements $e_1,e_2\in S$ such that $h_i(e_i)\geq \frac{1}{k}$ for $i\in\{1,2\}$, and we are done. Now assume that the property holds for $m\leq p$ and we show it for $m=p+1$ by contradiction. Consider an arbitrary set $X_0$ of size $p$, such that $h_i(X_0)\geq\frac{p-1}{k}$. Such a set exists by assumption, and note that if for some $i$, say $i=1$, $h_1(X_0)\geq \frac{p}{k}$, we are done, since we can add an element $e\in S$ to $X_0$ such that $h_2(e+X_0)\geq \frac{p}{k}$. So we assume that $\frac{p-1}{k}\leq h_i(X_0)<\frac{p}{k},\forall i$ to set up the contradiction.
		
		Now, consider the reduced ground set $S-X_0$. Then we have that $h_i(S-X_0)> 1-\frac{p}{k},\forall i$ and since $|S-X_0|=k-p$, we have for the reduced ground set $S-X_0$, that there exists a set $X_1$ of size $p$ such that $h_i(X_1)\geq \frac{p-1}{k-p}h_i(S-X_0)> \frac{p-1}{k}$, by the induction assumption. Using our second assumption (for contradiction), we have that $h_i(X_1)<\frac{p}{k}$. We repeat this until $|S-\cup_j X_j|\leq p$. Let $X'=S-\cup_j X_j$, then since $h_i(S)=1\,\forall i$, we have $0<|X'|=p'\leq p$ and $h_i(X')> \frac{p'}{k}\,\forall i$. 
		Now using the induction assumption for $m=p+1-p'$ and ground set $S-X'$, we have a set $Y$ with $|Y|= p+1-p'$, such that $h_i(X'\cup Y)\geq h_i(X')+\frac{p-p'}{k-p'}(1-h_i(X'))>\frac{p}{k}\,, \forall i$, yielding a contradiction.\qed
		  \end{proof}
	
	Based on this lemma, consider the below generalized greedy algorithm,
	
	{\centering
		\vspace{-3mm}
		\begin{minipage}{\textwidth}
			\begin{algorithm}[H]
				\caption{Generalized Greedy Algorithm}
				\label{gengred}
				\textbf{input: } $m,V_1,V_2$
				\begin{algorithmic}[1]
					\State Initialize $A=\emptyset$
					\While{$|A|< k$}
					\State $m= \min\{m,k-|A|\}$ 
					\State Find $\Big\{X|\, f_i(X|A)\geq \frac{m-1}{k}[V_i-f_i(A)], X\subseteq N-A, |X|=m\Big\}$  by enumeration
					\State $A=A\cup X$
					\EndWhile
					\State Output:$ A$
				\end{algorithmic}
			\end{algorithm}
		\end{minipage}
		
		\centering}
	
	\begin{thm}\label{nemgen}
		If there exists a set $S$, $|S|= k$, such that $f_i(S)\geq V_i$, then Algorithm \ref{gengred} finds a set $A$, $|A|= l\leq k$, such that, $$f_i(A)\geq \Big(\beta\Big(0,\frac{l}{k}\Big)-1/\Theta(m)\Big)V_i,$$ by making $O(n^{m+1})$ queries. 
	\end{thm}
	\begin{proof}
	The analysis closely follows that of Lemma \ref{nem} in \cite{nem,nem1} and Appendix \ref{app1}. Let $A_j$ be the set of size $m\times j$ after iteration $j$, then using Lemma \ref{pareto} for monotone submodular functions $f'_i(.)=f_i(.|A_j)$ with $V'_i=V_i-f_i(A_j)$ we have that $\exists X,|X|\leq m$, such that $f_i(X|A_j)\geq \frac{m-1}{k}[V_i-f_i(A_j)]$. Hence, $f_i(A_{j+1})\geq f_i(A_j)+\frac{m-1}{k}[V_i-f_i(A_j)]$. Now similar to what is shown in \cite{nem,nem1}, this gives us that after $l$ iterations, 
	\begin{eqnarray*}
	f_i(A_l)\geq \Big(1-\Big(1-\frac{m-1}{k}\Big)^{l}\Big)V_i&&=\beta\Big(0,\frac{ml}{k}\Big)(1-1/\Theta(m))V_i\\
	&&=\Big(\beta\Big(0,\frac{ml}{k}\Big)-1/\Theta(m)\Big)V_i.\qed
\end{eqnarray*}
	\end{proof}
	Thus, Algorithm \ref{gengred} gives a deterministic $(1-1/e)-1/\Theta(m)$ approximation for bi-objective maximization of monotone submodular functions subject to cardinality constraints. We define this problem more generally and formally in Section \ref{consttau}, but for now, note that the problem of maximizing the minimum of two monotone submodular functions is a special case of the above, where we don't need to input $V_i$, and hence for Algorithm \ref{a0.63}, we have that while in Phase 1, after $l$ iterations,
	\begin{eqnarray} 
	&&g(A)=\min\{f(A-a_1),f(A-a_2)\}\notag\\
	&&\geq\Big(\beta\Big(0,\frac{l}{k}\Big)-1/\Theta(m)\Big) \min\{f(OPT(k,N,1)-a_1),f(OPT(k,N,1)-a_2)\}\notag\\
	&&\geq \Big(\beta\Big(0,\frac{l}{k}\Big)-1/\Theta(m)\Big)g(OPT(k,N,1))\label{afew}
	\end{eqnarray}

	Note that for Algorithm \ref{gengred}, Lemma \ref{key} applies, albeit with the additive $-1/\Theta(m)$ term. We now present the analysis of Algorithm \ref{a0.63}.
	\begin{thm}\label{0.63}
		Given $m\geq 2$, Algorithm \ref{a0.63} is $\beta(0,\frac{k-2m-2}{k})-1/\Theta(m)\xrightarrow{k\to\infty}(1-1/e)-1/\Theta(m)$ approximate, and makes $O(n^{m+1})$ queries.
	\end{thm}
\begin{proof}
	Let $A_0=\{a_1,a_2\}$. Consider the function $g^0(S)=\min_{i\in\{1,2\}}\{f(S-a_i)\}$ and let $z$ be a minimizer of output set $A$ as usual and define $z^0(S)=\arg\min_{i\in\{1,2\}}\{f(S-a_i)\}$. 
	Note that if $z^0(S)\in \mathcal{Z}(S)$ then $g^0(S)=g(S)$. Also, with the standard definition of marginal, note that for any set $S\cap A_0=\emptyset$, $g^0(S|X)\geq \min_{i\in\{1,2\}}\{f(S|X-a_i)\}\geq f(S|X)$. Let $U$ be the set added during Phase 1 and similarly $W$ during Phase 2. Also, let $U=\{u_1,\dots,u_p\}$, where each $u_i$ is a set of size $m$ and  similarly $W=\{w_1,\dots,w_r\}$, where each $w_i$ is a singleton. Let $OPT=g(OPT(k,N,1))$ and note that if $r=0$ i.e., Phase 2 doesn't occur, we have that $z\in A_0$ and using (\ref{afew}), $g(A)\geq \big[\beta(0,\frac{k-2}{k-1})-1/\Theta(m)\big]OPT$. So assume $r>0$ and also $p\geq 2$, since the case $p=1$ will be easy to handle later. 
	Now, let $f(z|A-z)=\eta f(A)$, and so $g(A)=(1-\eta)f(A)$. Using analysis similar to Lemma \ref{key}, we will focus on showing that, 
	\begin{eqnarray}
	f(A)\geq \Big[\beta\Big(3\eta,\frac{k-2m}{k-1}\Big)-1/\Theta(m)\Big]OPT\label{tsh}
	\end{eqnarray}
	Lemma \ref{subkey} then gives $g(A)\geq \big[\beta(0,\frac{k-2m}{k-1})-1/\Theta(m)\big]OPT$.
	
	During the rest of the proof, we sometimes ignore the $1/\Theta(m)$ term with the understanding that it is present by default. To show (\ref{tsh}), let $a_i=z^0(A_0\cup U)$. Then, observe that, 
	\begin{eqnarray}
	&f(A)=& g^0(A_0\cup (U-u_p)) +g^0(u_p|A_0\cup (U-u_p))\notag\\
	&& +f(a_i|(A_0\cup U)-a_i) +f(W|A_0\cup U)\label{biggie}
	\end{eqnarray}
	For the first term in (\ref{biggie}), we have 
	$$g^0(A_0\cup (U-u_p))= g^0(A_0\cup u_1)+g^0(\{u_2,\dots,u_{p-1}\}|A_0\cup u_1).$$ Then, due to the greedy nature of Phase 1, we have using Theorem \ref{nemgen} (ignoring $1/\Theta(m)$ term), 
	\begin{eqnarray}
	g^0(\{u_2,\dots,u_{p-1}\}|A_0\cup u_1)\geq \beta\Big(0,\frac{mp-2m}{k-1}\Big)(OPT-g^0(A_0\cup u_1))\label{(b)}
	\end{eqnarray}
	As usual, the $k$ in the denominator was replaced by $k-1$ because we compare the value to a set of size $k-1$ ($OPT=f(OPT(k,N,1)-z)$). Similarly, for the last term in (\ref{biggie}), we have,
	\begin{eqnarray}
	f(W|A_0\cup U)\geq \beta\big(0,\frac{r}{k-1}\big)(OPT-f(A_0\cup U))\label{(a)}
	\end{eqnarray}
	Now we make some substitutions, let $\Delta= g^0(A_0\cup u_1)+g^0(u_p|A_0\cup (U-u_p))+f(a_i|(A_0\cup U)-a_i)$, $\alpha_p=\frac{(m-1)(p-2)}{k-1},\alpha_r=\frac{r}{k-1}$. Then using $k=mp+r+2$ we get, $\alpha_p+\alpha_r=\frac{k-2m}{k-1}-1/\Theta(m)$. Also, $f(A_0\cup U)=\Delta+g^0(\{u_2,\dots,u_{p-1}\}|A_0\cup u_1)$. This gives us,
	\begin{eqnarray*}
		&&f(A)=\Delta + g^0(\{u_2,\dots,u_{p-1}\}|A_0\cup u_1) +f(W|A_0\cup U) \\
		&&\overset{(a)}{\geq} \Delta + g^0(\{u_2,\dots,u_{p-1}\}|A_0\cup u_1) + \beta(0,\alpha_r)[OPT-f(A_0\cup U)]\\
		&&\geq \Delta+ (1-\beta(0,\alpha_r))g^0(\{u_2,\dots,u_{p-1}\}|A_0\cup u_1)+\beta(0,\alpha_r)(OPT-\Delta)\\
		&&\overset{(b)}{\geq} \Delta+(1/e^{\alpha_r})\beta(0,\alpha_p)(OPT-g^0(A_0\cup u_1))+\beta(0,\alpha_r)(OPT-\Delta)\\
		&&\geq \Delta+ (\beta(0,\alpha_r)+\beta(0,\alpha_p)/e^{\alpha_r})[OPT-\Delta]\\
		&&= \Delta+\Big[\beta\Big(0,\frac{k-2m}{k-1}\Big)-1/\Theta(m)\Big][OPT-\Delta]
	\end{eqnarray*}
	where the last equality holds asymptotically, $(a)$ comes from (\ref{(a)}) and $(b)$ from (\ref{(b)}).
	Assume for the time being, that for any $x$, $3f(x|A-x)\leq \Delta$, which implies that $\Delta\geq 3\eta f(A)$. Armed with these inequalities, the same simple algebra as in the proof of Lemma \ref{key}, gives us (\ref{tsh}). More concretely, ignoring the $1/\Theta(m)$ term, we have,
	\begin{eqnarray*}
		&f(A) &\geq \Delta+\beta\Big(0,\frac{k-2m}{k-1}\Big)[OPT-\Delta]\\
		&	&\geq 3\eta f(A) \Big(1-\beta\Big(0,\frac{k-2m}{k-1}\Big)\Big) +\beta\Big(0,\frac{k-2m}{k-1}\Big)OPT\\
			\end{eqnarray*}
			\vspace{-10mm}
				\begin{eqnarray*}
					 (1-3\eta e^{-\frac{k-2m}{k-1}}) f(A)&&\geq (1-e^{-\frac{k-2m}{k-1}}) OPT\\
			 		\implies f(A)&&\geq \beta\Big(3\eta,\frac{k-2m}{k-1}\Big)OPT
	\end{eqnarray*}
	
	To finish the proof, we need to show $3f(x|A-x)\leq \Delta$. We break this down by first showing in two steps that for all $x$, $2f(x|A-x)\leq  g^0(A_0\cup u_1)=f(a_2)+g^0(u_1|A_0)$, followed by proving that $f(x|A-x)\leq g^0(u_p|A_0\cup (U-u_p))+f(a_i|(A_0\cup U)-a_i)$.
	
	\textit{Step 1}: for all $x$, $f(x|A-x)\leq f(a_2)$. To see this for $x\neq a_1$, note that $f(x|A-x)\leq f(x|A_0-x)$ and further, $f(x|A_0-x)\leq f(a_2|a_1)\leq f(a_2)$, where the first inequality is because $a_2$ adds maximum marginal value to $a_1$. For $x=a_1$, since Phase 1 ends, we have that $f(a_1|A-a_1)\leq f(a_1|a_2+U)\leq f(y|(A_0\cup U)-y)$ for some $y$ not in $A_0$ and then we have $f(y|(A_0\cup U)-y)\leq f(a_2)$.
	
	\textit{Step 2}: $f(x|A-x)\leq g^0(u_1|A_0)$. For $x\not\in A_0$, we have that $f(x|A-x)\leq f(x|A_0)\leq g^0(x|A_0)\leq g^0(u_1|A_0)$. For $x\in A_0$, from the fact that $A_0\cup U$ has a minimizer $y\not \in A_0$, we have that $f(x|A-x)\leq f(x|(A_0\cup U)-x)\leq f(y|(A_0\cup U)-y)$ and further $f(y|(A_0\cup U)-y)\leq f(y|A_0)\leq g^0(y|A_0)\leq g^0(u_1|A_0)$.
	
	Finally, we show that $f(x|A-x)\leq g^0(u_p|A_0\cup (U-u_p))+f(a_i|(A_0\cup U)-a_i)$, for all $x$. Let $a_j=z^0(A_0\cup (U-u_p))$ and observe that, 
	\begin{eqnarray*}
		&&g^0(A_0\cup U)-g^0(A_0\cup (U-u_p))+f(a_i|(A_0\cup U)-a_i)\\
		&&=f((A_0\cup U)-a_i)-f((A_0-a_j)\cup (U-u_p))+f(a_i|(A_0\cup U)-a_i)\\
		&&= f(A_0\cup U)-f((A_0-a_j)\cup (U-u_p))\\
		&& \geq f(a_j|(A_0-a_j) \cup (U-u_p))
	\end{eqnarray*} 
	Now, before adding $u_p$, we have that $g^0(.)=g(.)$ and in fact, $a_j$ is a minimizer of $A_0\cup (U-u_p)$, so clearly, for all $x\in A_0\cup (U-u_p)$, the desired is true. For $x\in u_p$, it is true since $f(u_p|A_0\cup (U-u_p))\leq g^0(u_p|A_0\cup (U-u_p))$. For $x\in W$, we have that $f(x|A-x)\leq f(x|A_0\cup (U-u_p))\leq g^0(x|A_0\cup (U-u_p))\leq g^0(u_p|A_0\cup (U-u_p)))$, and we are done.
	
	
	The case $p=1$ can be dealt with same as above, except that now $\Delta= g^0(A_0\cup u_1)+f(a_i|(A_0\cup u_1)-a_i)+f(w_1|A_0\cup u_1)=f(A_0\cup u_1\cup w_1)$. \qed
\end{proof}
	Before discussing a similar result for constant $\tau\geq 1$, we first describe a fast 0.387 approximation for $\tau=o(\sqrt{k})$, which we will use when showing a $(1-1/e)-\epsilon$ approximation for fixed $\tau$.
		
		\vspace{-1mm}
		\subsubsection{0.387 Algorithm for $\tau\ll \sqrt{k}$ in the possible absence of ``copies"}
		\vspace{-1mm}
		One can recast the first algorithm in Section \ref{copygen}, which greedily chooses $\{a_1,\dots,a_{k-2\tau^2}\}$ elements and adds $\tau$ copies for each of the first $2\tau$ elements. as greedily choosing $2\tau$ elements, ignoring them and choosing another $2\tau$ greedily (which will be copies of the first $2\tau$) and repeating this $\tau$ times in total, leading to a set which contains $A_{2\tau}$ and $\tau-1$ copies of each element in $A_{2\tau}$. Then, ignoring this set, we greedily add till we have $k$ elements in total. Thus, the algorithm essentially uses the greedy algorithm as a sub-routine $\tau+1$ times. Based on this idea, we now propose an algorithm for $\tau=o(\sqrt{k})$, which can also be viewed as an extension of the 0.387 algorithm for $\tau=1$. To be more precise, it achieves an asymptotic guarantee of 0.387 for $\tau=o(\sqrt{\frac{k}{c(k)}})$, where $c(k)$ is an input parameter that governs the trade off between how fast the guarantee approaches 0.387 as $k$ increases and how large $\tau$ can be for the guarantee to still hold. In fact, the guarantee is $0.387\Big(1-\frac{1}{\Theta(c(k))}\Big)$ with $c(k)$ being a function monotonically increasing in $k$ and approaching $\infty$ as $k \to \infty$. The factor also degrades proportionally to $1-\frac{\tau^2c(k)}{k}$, as $\tau$ approaches $\sqrt{\frac{k}{c(k)}}$. 
		
		{\centering
			\vspace{-5mm}
			\begin{minipage}{\textwidth}
				\begin{algorithm}[H]
					\caption{Algorithm for $\tau=o\Big(\sqrt{\frac{k}{ c(k)}}\Big)$}
					\label{greedytau}
					\begin{algorithmic}[1]
						\State Initialize $\tau'=c(k)\tau^2,A_0=A_1=X=\emptyset$.
						\While {$|A_0|<\tau'$}
						\While {$|X|<\tau'/\tau$} $X=X+\argmax{x\in N-(A_0\cup X)}f(x|X)$
						\EndWhile
						\State $A_0= A_0 \cup X$; $X=\emptyset$
						\EndWhile
						\While{$|A_1|< k-\tau'$} $A_1=A_1+\argmax{x\in N-(A_0\cup A_1)}f(x|A_1)$
						\EndWhile
						\vspace{-3mm}
						\State Output: $A_0\cup A_1$
					\end{algorithmic}
				\end{algorithm}
			\end{minipage}
			\vspace{-2mm}
			\centering}
		

		\begin{thm}\label{genanal}
			Algorithm \ref{greedytau} has an approximation ratio of $\frac{e-1}{2e-1+\frac{e-1}{c(k)}}=\frac{e-1}{2e-1}\big(1-\frac{1}{\Theta(c(k))}\big)\xrightarrow{k\to \infty} 0.387$ for $\tau=o\big(\sqrt{\frac{k}{c(k)}}\big)$.
		\end{thm}
		\begin{proof}
			Let $A=A_0\cup A_1$ be the output with $A_0,A_1$ as in the algorithm. Define $Z_0=A_0\cap Z$ and $Z_1=Z-Z_0=A_1\cap Z$. Let $OPT(k-\tau,N-Z_0,0)=A'_0\cup X$ where $A'_0=OPT(k-\tau,N-Z_0,0)\cap A_0$ and $X\cap A'_0=\emptyset$. Now note that,
			\begin{eqnarray*}
				f(A_0-Z_0)+f(OPT(k-\tau,N-A_0,0))&&\geq f(A'_0)+f(X)\notag\\
				&&\geq f(OPT(k-\tau,N-Z_0,0))\notag\\
				&&\geq g(OPT(k,N,\tau))\quad [\because \text{Lemma }\ref{f and g}]\notag
			\end{eqnarray*}
			Which implies,
			\begin{eqnarray}
			f(A_0-Z_0)\geq g(OPT(k,N,\tau))  -f(OPT(k-\tau,N-A_0,0))\label{first}
			\end{eqnarray}
			In addition,\begin{eqnarray}f(A_1)\geq \beta\Big(0,\frac{k-\tau'}{k-\tau}\Big)f(OPT(k-\tau,N-A_0,0))\label{second}\end{eqnarray}
			Next, index disjoint subsets of $A_0$ based on the loop during which they were added. So the subset added during loop $i$ is denoted by $A^i_0$, where $i\in\{1,\dots,\tau\}$. So the last subset consisting of $\tau c(k) $ elements is $A^{\tau}_0$. 
			
			Now, if $Z_0$ includes at least one element from each $A^i_0$ then $Z_1=\emptyset$ and for this case we have from (\ref{first}) and (\ref{second}) above,
			\begin{eqnarray*}
				f(A-Z)&&\geq \max\{f(A_0-Z_0),f(A_1)\}\\
				&&\geq \max\{g(OPT(k,N,\tau))-f(OPT(k-\tau,N-A_0,0)), \\
				&&\quad\quad\quad\beta\Big(0,\frac{k-\tau'}{k-\tau}\Big)f(OPT(k-\tau,N-A_0,0))\}\\
				&&\geq \frac{\beta\Big(0,\frac{k-\tau'}{k-\tau}\Big)}{1+\beta\Big(0,\frac{k-\tau'}{k-\tau}\Big)} g(OPT(k,N,\tau))\\
				&&\xrightarrow{k\to \infty}\frac{e-1}{2e-1}g(OPT(k,N,\tau))
			\end{eqnarray*} 
			Next, suppose that $|Z_1|> 0$, then there is some $A^j_0$ such that $A^j_0\cap Z=\emptyset$. Further let $f(Z_1|A-Z)=\eta f(A_1)$, then since $|A^j_0|\geq c(k)|Z_1|$, similar to (\ref{iter}), we have due to greedy iterations and submodularity, $f(A^j_0)\geq c(k)f(Z_1|A-Z)=c(k) \eta f(A_1)$. Also note that,
			\begin{eqnarray}
			f(A-Z)\geq f(A_1-Z_1)\geq f(A_1)-f(Z_1|A-Z)\geq (1-\eta)f(A_1)\label{third}
			\end{eqnarray} 
%
			Moreover, let $A'_1$ be the set of first $\tau$ elements of $A_1$. Then, due to greedy iterations we have $f(A'_1)\geq f(Z_1|A_1-Z_1)\geq\eta f(A_1)$. Thus, from Lemma \ref{key}, with $N$ replaced by $N-A_0$, $k$ replaced by $k-\tau$, $S=A'_1$ with $c=\eta$ and $l=k-|S|=k-\tau'-\tau$, we have,
			\begin{eqnarray}f(A_1)\geq \beta(\eta,\frac{k-\tau'-\tau}{k-\tau})f(OPT(k-\tau,N-A_0,0))\label{fourth}\end{eqnarray}
			From (\ref{first}), (\ref{third}) and (\ref{fourth}),
			\begin{eqnarray*}
				f(A-Z)&&\geq \max\{f(A_0-Z_0),(1-\eta)f(A_1),f(A^j_0)\}\\
				&&\geq \max\{f(A_0-Z_0),(1-\eta)f(A_1), c(k)\eta f(A_1)\}\\
				&&\overset{(a)}{\geq} \max\{g(OPT(k,N,\tau))-f(OPT(k-\tau,N-A_0,0)),\\
				&&\quad \quad\quad  \frac{c(k)}{1+c(k)} \beta(\frac{1}{1+c(k)},\frac{k-\tau'-\tau}{k-\tau})f(OPT(k-\tau,N-A_0,0))\}\\
				&&\geq \frac{\frac{c(k)}{1+c(k)}\beta\Big(\frac{1}{1+c(k)},\frac{k-\tau'-\tau}{k-\tau}\Big)}{1+\frac{c(k)}{1+c(k)}\beta\Big(\frac{1}{1+c(k)},\frac{k-\tau'-\tau}{k-\tau}\Big)} g(OPT(k,N,\tau))\\
				&&\xrightarrow{k\to \infty}\frac{e-1}{2e-1}g(OPT(k,N,\tau))
			\end{eqnarray*} 
			where $(a)$ follows by substituting $\eta=\frac{1}{1+c(k)}$.\qed
			\end{proof}
		\subsubsection{$(1-1/e)-\epsilon$ algorithm for constant $\tau$}\label{consttau}
		In Section \ref{gensec}, inspired by the 2-Copy algorithm, we derived a $(1-1/e)-1/\Theta(m)$ approximation for the general case, using a phase wise approach. In the first phase, while the minimizers are restricted to the set $A_0=\{a_1,a_2\}$, we build a set robust to removal of either of these elements by deriving and using an algorithm for bi-objective maximization of monotone submodular functions. In the second and final phase, we filled in the rest of the set with standard greedy iterations (like Algorithm \ref{greedy}). 
		
		This phase wise approach however, doesn't generalize well for $\tau>1$ since we can have a minimizer that intersects with the initial set but is not a subset of the initial set, unlike for $\tau=1$ where a minimizer $z$ is either in $\{a_1,a_2\}$ or not. Instead, an alternative approach comes from reinterpreting the result from $\tau=1$, where, we want to build a set that has a large value on both $f_1(.)=f(.|a_1)$ and $f_2(.)=f(.|a_2)$ simultaneously, to deal with the scenarios when either of these elements is a minimizer. Also, we can capture the notion of continuing greedily w.r.t. $f(.)$ once the set becomes robust to removal of either of $a_1$ or $a_2$, by considering a third function $f_3(.)=f(.|\{a_1,a_2\})$. Thus, instead of separate phases, we can think about a single multi-objective problem over three monotone submodular functions $f_1,f_2,f_3$, and try to add a set $A_1$ to $A_0=\{a_1,a_2\}$, such that $f_i(A_1)\geq (1-1/e)f_i(OPT(k,N,1)), \forall i$. To see why this serves our purposes, consider the scenario where $a_2$ is a minimizer for the final set $A$,
		 \begin{eqnarray*}
		 	&g(A)=f(A_1+A_0-a_2)&=f_1(A_1)+f(a_1)\\
		 	&&\geq f(a_1)+(1-1/e)(f(OPT(k,N,1))-f(a_1))\\
		 	&&\geq (1-1/e)g(OPT(k,N,1)).
		 	\end{eqnarray*}
		  Generalizing this for larger $\tau$, we start with the set $A_0$ obtained by running Algorithm \ref{greedytau} with $\tau'=3\tau^2$, implying $|A_0|=3\tau^2$. 
		Now, consider the monotone submodular functions $f_i(.)=f(.|Y_i)$ for every possible subset $Y_i$ ($|Y_i|\geq 3\tau^2-\tau$) of $A_{0}$ and denote the set of functions by $\mathcal{L}$ ($|\mathcal{L}|=2^{\tilde{\Theta}(\tau)}$, large but constant). 
		
		Assuming there exists a set $S$ of size $k-3\tau^2$, such that $f_i(S)\geq (1-\Theta(\frac{1}{k})) \big[g(OPT(k,N,\tau))-f(Y_i)\big]$, we would like to solve an instance of a multi-objective submodular maximization problem (made more precise later) to find a set $A_1$ of size $k-3\tau^2$ such that $f_i(A_1)\geq \beta(0,1)
		(1-\Theta(\frac{1}{k})) \big[g(OPT(k,N,\tau))-f(Y_i)\big]$. We know that $OPT(k,N,\tau)$ is a set such that, $f_i(OPT(k,N,\tau))\geq f(OPT(k,N,\tau))-f(Y_i)\geq g(OPT(k,N,\tau))-f(Y_i)$, however it is a set of size $k$. So before we can puzzle out how to find set $A_1$, we need the proof of existence of set $S$. 
		\begin{lem}
			Given a constant number of monotone submodular functions $\mathcal{L}=\{f_i,\dots,f_l\}$ and values $\mathcal{V}=\{V_i,\dots,V_l\}$ and a set $S'$ of size $k$, such that $f_i(S')\geq V_i,\forall i$. There exists a set $S$ of size $k-p$ such that,
			 $$f_i(S)\geq \Big(1-\frac{l}{k-(p+l-1)}\Big)^{p} V_i=\Big(1-\Theta(1/k)\Big)V_i$$ for constant $l,p\ll k$.  
		\end{lem} 
		\begin{proof}
			This is clearly true for $l=1$ since $\exists S\subset S', |S|=k-p$, such that, $f(S)\geq \frac{k-p}{k}f(S')$. More generally, we show that $\exists S_1\subset S',|S_1|=k-1$, such that $f_i(S_1)\geq \frac{k-2l}{k-l}f_i(S')$. Then we can reapply this on $S_1$ to get $S_2\subset S_1, |S_2|=k-2$ and $f_i(S_2)\geq \frac{k-2l-1}{k-l-1}f_i(S_1)$. Repeating this $p$ times over all, we get $S_p\subset S',|S_p|=k-p$ and $f_i(S_p)\geq \Pi_{j=0}^{p-1} \frac{k-j-2l}{k-j-l}V_i\geq \big(\frac{k-p-2l+1}{k-p-l+1}\big)^pV_i$, which gives the desired. So, it remains to show the claimed lower bound on $f_i(S_1)$. 
			
			Let $S'=\{s'_1,\dots,s'_k\}$ and let $S'_j=\{s'_1,\dots,s'_j\}$. Similar to the argument in Lemma \ref{pareto}, it suffices to show this for modular functions $h_i$ where $h_i(S')=1$ and  $h_i(s'_j)=f_i(s'_j|S'_{j-1}), \forall i$. W.l.o.g., assume  that $S'$ is indexed such that $h_1(s'_j)\leq h_1(s'_{j+1})$. Consider the set $S'_{1+\ceil{\frac{l-1}{l}k}}$ of the $1+\ceil{\frac{l-1}{l}k}$ smallest elements elements w.r.t. $h_1$. There are at most $(k-\ceil{\frac{l-1}{l}k})(l-1)\leq \ceil{\frac{l-1}{l}k}< |S'_{1+\ceil{\frac{l-1}{l}k}}|$ distinct elements which are in the top $(k-\ceil{\frac{l-1}{l}k})=\floor{\frac{k}{l}}$ for some function $h_i, i\geq 2$. Hence $\exists j_0\leq 1+\ceil{\frac{l-1}{l}k}$ such that $s'_{j_0}$ is one of the $\ceil{\frac{l-1}{l}k}$ smallest elements for each function $h_i, i\geq 2$. 
			This implies that if we remove $s'_{j_0}$, we have $h_i(S'-s'_{j_0})\geq 1-\frac{1}{\floor{\frac{k}{l}}}\geq \frac{k-2l}{k-l},\forall i$.\qed
		\end{proof}
		Now that we know that set $S$ of size $k-3\tau^2$ satisfying the desired inequalities exists, we need an algorithm to find a set $A_1$ that $(1-1/e)$-approximates the value of $S$ for each $f_i$. While Algorithm \ref{gengred} only works for up to two functions, as mentioned in the beginning, a randomized $(1-1/e-\epsilon)$ algorithm for a more general problem was given in \cite{swap} and we use it here.
		\begin{lem}[Chekuri, Vondr{\'a}k, Zenklusen \cite{swap}]
			Given constant number of monotone submodular functions $\mathcal{L}=\{f_i,\dots,f_l\}$ and values $\mathcal{V}=\{V_i,\dots,V_l\}$. If there exists a set $S$ ($|S|\leq k$) with $f_i(S)\geq V_i,\forall i$, there is a polynomial time algorithm which finds a set $X$ of size $\leq k$ such that $f_i(X)\geq(1-1/e-\epsilon) V_i,\forall i$, with constant probability, for a constant $\epsilon\leq \frac{1}{\log l}$, making $O(n^{1/\epsilon^3})$ queries. If a set $S$ with the value lower bounds given by $\mathcal{V}$ doesn't exist, then the algorithm gives a certificate of non-existence. \end{lem}
		Denote the algorithm by $\mathcal{A}$ and let the set output be $\mathcal{A}(\mathcal{V},\mathcal{L})$, for inputs $\mathcal{V},\mathcal{L}$ as described above. To ease notation, from here on we generally ignore the $\epsilon$ term.

%
				
		A final hurdle in using the above is that we need to input the values $V_i=g(OPT(k,N,\tau))-f(Y_i)$ and hence, we need an estimate of $OPT=g(OPT(k,N,\tau))$. We can overestimate $OPT$ as long as the problem remains feasible. However, underestimating $OPT$ results in a loss in guarantee. 
		Using Algorithm \ref{greedytau}, we can quickly find lower and upper bounds  $lb,ub$ such that $lb\leq OPT\leq ub=lb/0.387$ and then run the multi-objective maximization algorithm above with a geometrically increasing sequence of $O(1/\log(1+\delta))$ many values to get an estimate $OPT'$ within factor $(1\pm\delta)$ of $OPT$.
		
		\emph{In summary, our scheme starts with the set $A_0$ of size $3\tau^2$, obtained by running Algorithm \ref{greedytau} with $\tau'=3\tau^2$, then uses the algorithm for multi-objective optimization as a subroutine to find the estimate $OPT'\geq (1-\delta) OPT$ and simultaneously a set $A_1=\mathcal{A}(\{OPT'-f(Y_i)\}_i,\mathcal{L})$ of size $k-3\tau^2$, such that $f_i(A_1)\geq (1-1/e)(1-\Theta(1/k))(OPT'-f(Y_i))$.} 
		
		\emph{The final output is $A=A_{0}\cup A_1$ and we next show that this is asymptotically $(1-1/e-\epsilon)$ approximate.}
		
		\begin{proof}
			We ignore the $\epsilon$ and $\Theta(1/k)$ terms to ease notation. First, note that if $Z\subseteq A_{0}$, then $f(.|A_0-Z)\in \mathcal{L}$ gives us, $f(A_1|A_0-Z)\geq(1-1/e) (OPT'-f(A_0-Z))$. Hence,
			\begin{eqnarray*}
			g(A)&&=f(A_0-Z)+f(A_1|A_0-Z)\\
			&&\geq f(A_0-Z)+(1-1/e) (OPT'-f(A_0-Z))\\
			&&\geq (1-1/e)OPT'. 
			\end{eqnarray*} 
			
			If $Z\not\subseteq A_{0}$, then let $Z_1=Z\cap A_1$ and $Z_0=Z-Z_1$. Similar to the proof of Theorem \ref{genanal}, let $A^i_0$ denote the $i$th set of $3\tau$ elements greedily chosen for constructing $A_0$, $i\leq \tau$. Since $|Z_0|<\tau$, $\exists i$ such that $A^i_0\cap Z_0=\emptyset$. Then analogous to the proof of Theorem \ref{genanal}, we have due to greedy additions, 
			\begin{eqnarray}
			f(A_0-Z_0)\geq f(A^i_0)\geq 3f(Z_1|A-Z)\label{rnow}
			\end{eqnarray}
			 (in contrast with $c(k)f(Z_1|A-Z)$ in Theorem \ref{genanal}). Now, since $f(.|A_{0}-Z_0)$ is one of the functions in $\mathcal{L}$, we have $f(A_1|A_0-Z_0)\geq(1-1/e) (OPT'-f(A_{0}-Z_0))$ which implies, $$f(A-Z_0)=f(A_0-Z_0)+f(A_1|A_0-Z_0)\geq f(A_{0}-Z_0)+(1-1/e) (OPT'-f(A_{0}-Z_0)).$$ Further, letting $f(Z_1|A-Z)=\eta f(A-Z_0)$ and using (\ref{rnow}),
			\begin{eqnarray*}
				 f(A-Z_0)&&\geq \frac{3\eta}{e} f(A-Z_0)+ (1-1/e)OPT'\\
				 &&\geq \beta(3\eta,1)OPT'.
			\end{eqnarray*}
			Now, using Lemma \ref{subkey} we have $g(A)=f(A-Z)\geq (1-\eta)f(A-Z_0)\geq \beta(0,1)OPT'$. \qed
			
		\end{proof}
	Finally, consider the following generalization of Lemma \ref{afew},
			\begin{cnj}\label{combi}
				Given $l\geq 1$ (treated as a constant) monotone submodular functions $f_1,\dots,f_l$ on ground set $N$, a set $S\subseteq N$ of size $k$, such that $f_i(S)\geq V_i$ for all $i\in[l]$ and an arbitrary $m$ with $k\geq m\geq l$, there exists a set $X\subseteq S$ of size $ m$, such that, $$f_i(X)\geq \frac{m-\Theta(1)}{k}V_i,\quad \forall i\in[l]$$ 
			\end{cnj}
			If true, the above would give a greedy deterministic $(1-1/e)-1/\Theta(m)$ approximation for the multi-objective optimization problem and thus also for our robust problem, for constant $\tau$.
	
	\vspace{-5mm}
	
	\section{Extension to general constraints}
	\vspace{-2mm}
	So far, we have looked at a robust formulation of $P1$, where we have a cardinality constraint. However, there are more sophisticated  applications where we find instances of budget or even matroid constraints. 
	In particular, consider the generalization $\underset{A\in \mathcal{I}}{\max}\quad \underset{|B|\leq\tau}{\min}\, f(A\backslash B)$, for some independence system $\mathcal{I}$. By definition, for any feasible set $A\in \mathcal{I}$, all subsets of the form $A\backslash B$ are feasible as well, so the formulation is sensible. Let's briefly discuss the case of $\tau=1$
	and suppose that we are given an $\alpha$ approximation algorithm $\mathcal{A}$, with query/run time $O(R)$ for the $\tau=0$ case. Let $G_0$ denote its output and $z_0$ be a minimizer of $G_0$.  Consider the restricted system $\mathcal{I}_{z_0}=\{A:z_0\in A, A\in \mathcal{I}\}$. Now, in order to be able to pick elements that have small marginal on $z_0$ but large value otherwise, we can generalize the notion of ignoring $z_0$ by maximizing the monotone submodular function $f(.\backslash z_0)$ subject to the independence system $I_{z_0}$. However, unlike the cardinality constraint case, where this algorithm gives a guarantee of 0.387, the algorithm can be arbitrarily bad in general (because of severely restricted $I_{z_0}$, for instance). We tackle this issue by adopting an enumerative procedure.
	%
	
	Let $\mathcal{A}_j$ denote the algorithm for $\tau=j$ and let $\mathcal{A}_j(N,Z)$ denote the output of $\mathcal{A}_j$ on ground set $N$ and subject to restricted system $\mathcal{I}_Z$. 
	Finally, let $\hat{z}(A)=\argmax{x\in A} f(x)$. With this, we have for general constraints:
	
	{\centering
		\vspace{-3mm}
		\begin{minipage}{\textwidth}
			\begin{algorithm}[H]
				\caption{$\mathcal{A}_{\tau}: \frac{\alpha}{\tau+1}$ for General Constraints}
				\label{gen2a}
				\begin{algorithmic}[1]
					\State Initialize $i=0,Z=\emptyset$
					\While {$N-Z\neq \emptyset$}
					\State $G_i=\mathcal{A}_0(N-Z,\emptyset)$
					\State $z_i\in \hat{z}(G_i)$; \quad $Z=Z\cup z_i$
					\State $M_i=z_i\cup \mathcal{A}_{\tau-1}(N-Z,z_i)$; \quad $i=i+1$
					\EndWhile
					\vspace{-0.5mm}
					\State Output: $\text{argmax}\{ g_{\tau}(S)|S\in\{G_j\}_{j=0}^{i} \cup \{M_j\}_{j=0}^{i}\}$
				\end{algorithmic}
			\end{algorithm}
		\end{minipage}
		\centering}
	
	To understand the basic idea behind the algorithm, assume that $z_0$ is in an optimal solution for the given $\tau$. Then, given the algorithm $\mathcal{A}_{\tau-1}$, if a minimizer of the set $M_0=z_0 \cup \mathcal{A}_{\tau-1}(N-z_0,z_0)$ includes $z_0$, it only removes $\tau-1$ elements from $\mathcal{A}_{\tau-1}(N-z_0,z_0)$. On the other hand, if a minimizer doesn't include $z_0$,  $g_{\tau}(M_0)\geq f(z_0)\geq \frac{f(M_0)-g_{\tau}(M_0)}{\tau}$. These two cases yield the desired ratio, however, since $z_0$ need not be in an optimal solution, we systematically enumerate.

	\begin{thm}\label{genlem}
		Given an $\alpha$ approximation algorithm $\mathcal{A}$ for $\tau=0$ with query time $O(R)$, algorithm $\mathcal{A}_{\tau}$ described above guarantees ratio $\frac{\alpha}{\tau+1}$ for general $\tau$ with query time $O(n^{\tau} R+n^{\tau+1})$
	\end{thm}
	\begin{proof}
		We proceed via induction on $j\in\{0,\dots,\tau\}$. Clearly, for $j=0$, $\mathcal{A}_0\equiv \mathcal{A}$, and the statement holds. Assume true for $j\in\{0,1,\dots,\tau-1\}$, then we show validity of the claim for $\mathcal{A}_{\tau}$. The query time claim follows easily since the while loop runs for at most $n$ iterations and each iteration makes $O(n^{\tau}+n^{\tau-1}R)$ queries (by assumption on query time of $\mathcal{A}_{\tau-1}$) and updating the best solution at the end of each iteration (counts towards the final output step) takes $O(n^{\tau})$ time to find the minimizer by brute force for two sets $G_i$ and $M_i$.
		
		Now, let $OPT(\mathcal{I},N,\tau)$ denote an optimal solution to $\max_{A\in \mathcal{I}}\min_{|B|=\tau}f(A-B)$ on ground set $N$ and assume that $z_0\in OPT(\mathcal{I},N,\tau)$. For any minimizer $B$ of $A$, we have for every element $z\in\hat{z}(A)$, $f(z)\geq \frac{f(B)}{\tau}\geq \frac{f(A)-f(A-B)}{\tau}$. Let $Z_0$ denote a minimizer of $G_0$. Hence, if $z_0\not\in Z_0$, we have that $g_{\tau}(G_0)\geq f(z_0)\geq \frac{f(G_0)-g_{\tau}(G_0)}{\tau}$, 
		giving us $g_{\tau}(G_0)\geq \frac{f(G_0)}{\tau+1}$. 
		Instead if $z_0$ is in the minimizer of $G_0$ and if $f(G_0-Z_0)< \frac{f(G_0)}{\tau+1}$, then we have that $f(Z_0|G_0-Z_0)\geq \frac{\tau}{\tau+1}f(G_0)$, implying that $f(z_0)\geq \frac{f(G_0)}{\tau+1}$. Now, let $Z'_0$ denote the minimizer of $M_0$ and note that if $z_0\not\in Z'_0$, we are done. Else, we have that, 
		\begin{eqnarray*}
			g_{\tau}(M_0)=g_{\tau-1}(M_0-z_0)&&\geq \frac{\alpha}{\tau} g_{\tau-1}(OPT(\mathcal{I}_{z_0},N-z_0,\tau-1))\\
			&&\geq \frac{\alpha}{\tau} g_{\tau-1}(OPT(\mathcal{I},N,\tau)-z_0)\\
			&&\geq \frac{\alpha}{\tau} g_{\tau}(OPT(\mathcal{I},N,\tau))> \frac{\alpha}{\tau+1} g_{\tau}(OPT(\mathcal{I},N,\tau))
		\end{eqnarray*}
		Where the first inequality stems from the induction assumption, the second and third by our assumption on $z_0$. This was all true under the assumption that $z_0\in OPT(\mathcal{I},N,\tau)$, if that is not the case, we remove $z_0$ from the ground set and repeat the same process. The algorithm takes the best set out of all the ones generated, and hence there exists some iteration $l$ such that $z_l\in OPT(\mathcal{I},N,\tau)$ and analyzing that iteration as we did above, gives us the desired.    
		
		Finally, for the cardinality constraint case, we can avoid enumeration altogether and the simplified algorithm has runtime polynomial in $(n,\tau)$ and guarantee that scales as $\frac{1}{\tau}$, which for $\Omega(\sqrt{k})\leq\tau=o(k)$, is a better guarantee than the na\"{\i}ve one of $\frac{1}{k-\tau}$ from Section \ref{ggreed}.\qed
		%
	\end{proof}
	\section{Conclusion, Open Problems and Further Work} 
	\vspace{-2mm}
	We looked at a robust version of the classical monotone submodular function maximization problem, where we want sets that are robust to the removal of any $\tau$ elements. We introduced the special, yet insightful case of copies, for which we gave a fast and asymptotically $(1-1/e)$ approximate algorithm for $\tau=o(k)$. 
	
	For the general case, where we may not have copies, we gave a deterministic asymptotically $(1-1/e-1/\Theta(m))$ algorithm for $\tau=1$, with the runtime scaling as $n^{m+1}$. As a byproduct, we also developed a deterministic $(1-1/e)-1/\Theta(m)$ approximation algorithm for bi-objective monotone submodular maximization, subject to cardinality constraint. For larger but constant $\tau$, we gave a randomized $(1-1/e)-\epsilon$ approximation and conjectured that this could be made deterministic. Additionally, we also gave a fast and practical 0.387 algorithm for $\tau=o(\sqrt{k})$. Note that here, unlike in the special case of copies, we could not tune the algorithm to work for larger $\tau$ and in fact, there has been further work in this direction since the appearance of the conference and arXiv version of this paper. Notably, \cite{bog} generalizes the notion of geometrically reducing the number of copies from Section \ref{copygen}, and achieves a 0.387 approximation for $\tau=o(k)$. It is still open whether one can go all the way up to $(1-1/e)$ and no constant factor approximation or inapproximability result is known for $\tau=\Theta(k)$. 
	
	Finally, similar robustness versions can be considered for maximization subject to independence system constraints and we gave an enumerative black box approach that leads to an $\frac{\alpha}{\tau+1}$ approximation algorithm with query time scaling as $n^{\tau+1}$, given an $\alpha$ approximation algorithm for the non-robust case. 
	\vspace{-5mm}
	\section*{Acknowledgement}
	\vspace{-3mm}
	The authors would like to thank all the anonymous reviewers for their useful suggestions and comments on all the versions of the paper so far. In addition, 
	RU would also like to thank Jan Vondr{\'a}k for a useful discussion and pointing out a relevant result in \cite{swap}.
	\vspace{-2mm} 
	\bibliographystyle{plain}
	\bibliography{ipco}

\begin{thebibliography}{10}

\bibitem{fast}
A.~Badanidiyuru and J.~Vondr{\'a}k.
\newblock Fast algorithms for maximizing submodular functions.
\newblock In {\em SODA '14}, pages 1497--1514. SIAM, 2014.

\bibitem{book}
A.~Ben-Tal, L.~El~Ghaoui, and A.~Nemirovski.
\newblock {\em Robust optimization}.
\newblock Princeton University Press, 2009.

\bibitem{bert}
D.~Bertsimas, D.~Brown, and C.~Caramanis.
\newblock Theory and applications of robust optimization.
\newblock {\em SIAM review}, 53(3):464--501, 2011.

\bibitem{sim}
D.~Bertsimas and M.~Sim.
\newblock Robust discrete optimization and network flows.
\newblock {\em Mathematical programming}, 98(1-3):49--71, 2003.

\bibitem{bertsim}
D.~Bertsimas and M.~Sim.
\newblock The price of robustness.
\newblock {\em Operations research}, 52(1):35--53, 2004.

\bibitem{bog}
I.~Bogunovic, S.~Mitrovic, J.~Scarlett, and V.~Cevher.
\newblock Robust submodular maximization: A non-uniform partitioning approach.
\newblock In {\em ICML}, 2017.

\bibitem{deran}
N.~Buchbinder and M.~Feldman.
\newblock Deterministic algorithms for submodular maximization problems.
\newblock {\em CoRR}, abs/1508.02157, 2015.

\bibitem{buch}
N.~Buchbinder, M.~Feldman, J.S. Naor, and R.~Schwartz.
\newblock A tight linear time (1/2)-approximation for unconstrained submodular
  maximization.
\newblock FOCS '12, pages 649--658, 2012.

\bibitem{vond1}
G.~Calinescu, C.~Chekuri, M.~P{\'a}l, and J.~Vondr{\'a}k.
\newblock Maximizing a monotone submodular function subject to a matroid
  constraint.
\newblock {\em SIAM Journal on Computing}, 40(6):1740--1766, 2011.

\bibitem{swap}
C.~Chekuri, J.~Vondr{\'a}k, and R.~Zenklusen.
\newblock Dependent randomized rounding via exchange properties of
  combinatorial structures.
\newblock In {\em FOCS 10}, pages 575--584. IEEE, 2010.

\bibitem{hard2}
S.~Dobzinski and J.~Vondr{\'a}k.
\newblock From query complexity to computational complexity.
\newblock In {\em STOC '12}, pages 1107--1116. ACM, 2012.

\bibitem{hard}
U.~Feige.
\newblock A threshold of ln n for approximating set cover.
\newblock {\em Journal of the ACM (JACM)}, 45(4):634--652, 1998.

\bibitem{feige}
U.~Feige, V.S. Mirrokni, and J.~Vondrak.
\newblock Maximizing non-monotone submodular functions.
\newblock {\em SIAM Journal on Computing}, 40(4):1133--1153, 2011.

\bibitem{multiuni}
M.~Feldman, J.S. Naor, and R.~Schwartz.
\newblock A unified continuous greedy algorithm for submodular maximization.
\newblock In {\em FOCS '11}, pages 570--579. IEEE.

\bibitem{feld}
M.~Feldman, J.S. Naor, and R.~Schwartz.
\newblock Nonmonotone submodular maximization via a structural continuous
  greedy algorithm.
\newblock In {\em Automata, Languages and Programming}, pages 342--353.
  Springer, 2011.

\bibitem{gharan}
S.O. Gharan and J.~Vondr{\'a}k.
\newblock Submodular maximization by simulated annealing.
\newblock In {\em SODA '11}, pages 1098--1116. SIAM.

\bibitem{ml}
A.~Globerson and S.~Roweis.
\newblock Nightmare at test time: robust learning by feature deletion.
\newblock In {\em Proceedings of the 23rd international conference on Machine
  learning}, pages 353--360. ACM, 2006.

\bibitem{krau}
D.~Golovin and A.~Krause.
\newblock Adaptive submodularity: Theory and applications in active learning
  and stochastic optimization.
\newblock {\em J. Artificial Intelligence Research}, 2011.

\bibitem{sense1}
C.~Guestrin, A.~Krause, and A.P. Singh.
\newblock Near-optimal sensor placements in gaussian processes.
\newblock In {\em Proceedings of the 22nd international conference on Machine
  learning}, pages 265--272. ACM, 2005.

\bibitem{min2}
S.~Iwata, L.~Fleischer, and S.~Fujishige.
\newblock A combinatorial strongly polynomial algorithm for minimizing
  submodular functions.
\newblock {\em Journal of the ACM (JACM)}, 48(4):761--777, 2001.

\bibitem{sense2}
A.~Krause, C.~Guestrin, A.~Gupta, and J.~Kleinberg.
\newblock Near-optimal sensor placements: Maximizing information while
  minimizing communication cost.
\newblock In {\em Proceedings of the 5th international conference on
  Information processing in sensor networks}, pages 2--10. ACM, 2006.

\bibitem{main}
A.~Krause, H~B. McMahan, C.~Guestrin, and A.~Gupta.
\newblock Robust submodular observation selection.
\newblock {\em Journal of Machine Learning Research}, 9:2761--2801, 2008.

\bibitem{sense3}
J.~Leskovec, A.~Krause, C.~Guestrin, C.~Faloutsos, J.~VanBriesen, and
  N.~Glance.
\newblock Cost-effective outbreak detection in networks.
\newblock In {\em Proceedings of the 13th ACM SIGKDD international conference
  on Knowledge discovery and data mining}, pages 420--429. ACM, 2007.

\bibitem{feature2}
Y.~Liu, K.~Wei, K.~Kirchhoff, Y.~Song, and J.~Bilmes.
\newblock Submodular feature selection for high-dimensional acoustic score
  spaces.
\newblock In {\em Acoustics, Speech and Signal Processing (ICASSP), 2013 IEEE
  International Conference on}, pages 7184--7188. IEEE, 2013.

\bibitem{nem1}
G.L. Nemhauser and L.A. Wolsey.
\newblock Best algorithms for approximating the maximum of a submodular set
  function.
\newblock {\em Mathematics of operations research}, 3(3):177--188, 1978.

\bibitem{nem}
G.L. Nemhauser, L.A. Wolsey, and M.L. Fisher.
\newblock An analysis of approximations for maximizing submodular set
  functions—i.
\newblock {\em Mathematical Programming}, 14(1):265--294, 1978.

\bibitem{min1}
A.~Schrijver.
\newblock A combinatorial algorithm minimizing submodular functions in strongly
  polynomial time.
\newblock {\em Journal of Combinatorial Theory, Series B}, 80(2):346--355,
  2000.

\bibitem{denko}
M.~Sviridenko.
\newblock A note on maximizing a submodular set function subject to a knapsack
  constraint.
\newblock {\em Operations Research Letters}, 32(1):41--43, 2004.

\bibitem{feature1}
M.~Thoma, H.~Cheng, A.~Gretton, J.~Han, HP. Kriegel, A.J. Smola, L.~Song, S.Y.
  Philip, X.~Yan, and K.M. Borgwardt.
\newblock Near-optimal supervised feature selection among frequent subgraphs.
\newblock In {\em SDM}, pages 1076--1087. SIAM, 2009.

\bibitem{welf}
J.~Vondr{\'a}k.
\newblock Optimal approximation for the submodular welfare problem in the value
  oracle model.
\newblock In {\em STOC '08}, pages 67--74. ACM.

\bibitem{hard1}
J.~Vondr{\'a}k.
\newblock Symmetry and approximability of submodular maximization problems.
\newblock {\em SIAM Journal on Computing}, 42(1):265--304, 2013.

\bibitem{vond}
J.~Vondr{\'a}k, C.~Chekuri, and R.~Zenklusen.
\newblock Submodular function maximization via the multilinear relaxation and
  contention resolution schemes.
\newblock In {\em STOC '11}, pages 783--792. ACM, 2011.

\end{thebibliography}
	
	%
	%
	%
	%
	%
	\appendix 
	  \section{Appendix} 
	  \subsection{}\label{app1} 
	\begin{lem}[Nemhauser, Wolsey \cite{nem,nem1}]
		 For all $\alpha\geq 0$, greedy algorithm terminated after $\alpha k$ steps yields a set $A$ with $f(A)\geq \beta(0,\alpha)f(OPT(k,N,0))$.
	\end{lem}
	\begin{proof} 
	Let $A_i$ be the set at iteration $i$ of the greedy algorithm. Then by monotonicity, we have,
	$$f(A_i\cup OPT(k,N,0))\geq f(OPT(k,N,0))$$
	and by submodularity, 
	$$\sum_{e\in OPT(k,N,0)-A_i} f(e|A_i)\geq f(OPT(k,N,0)|A_i)\geq f(OPT(N,k,0))-f(A_i)$$
	 Hence, there exists an element $e$ in $OPT(k,N,0)-A_i$ such that,
	$$f(e|A_i)\geq (f(OPT(k,N,0))-f(A_i))/k$$
	Hence, we get the recurring inequality, 
	$$f(A_{i+1})\geq f(A_{i}+e)\geq f(A_i)+(f(OPT(k,N,0))-f(A_i))/k$$
	The above implies that the difference between the value of the greedy set and optimal solution decreases by a factor of $(1-1/k)$ at each step, so after $\alpha k$ steps, 
	\begin{eqnarray*}
	f(A_{\alpha k})&&\geq (1-(1-1/k)^{\alpha k})f(OPT(k,N,0))\\
	\implies f(A_{\alpha k})&&\geq \beta(0,\alpha) f(OPT(k,N,0))
	\end{eqnarray*}
	\end{proof}
	\subsection{}\label{app1.1}
		\begin{lem}
			There exists no polytime algorithm with approximation ratio greater than $(1-1/e)$ for $P2$ unless $P=NP$. For the value oracle model, we have the same threshold, but for algorithms that make only a polynomial number of queries.
		\end{lem}
		\begin{proof}
			We will give a strict reduction from the classical problem $P1$ (for which the above hardness result holds \cite{nem,hard}) to the robust problem $P2$. Consider an instance of $P1$, denoted by $(k,N,0)$. We intend to reduce this to an instance of $P2$ on an augmented ground set $N\cup X$ i.e.\ $(k+\tau,N\cup X,\tau)$.
			
			The set $X=\{x_1,\cdots,x_\tau\}$ is such that $f(x_i)=(k+1)  f(a_1)$ (and recall that $f(a_1)\geq f(a_i),\forall a_i\in N$) and $f(x_i|S)=f(x_i)$ for every $i$ and $S\subset N\cup X$ not containing $x_i$. We will show that $g(OPT(k+\tau,N\cup X,\tau))=f(OPT(k,N,0))$. 
			
			First, note that for an arbitrary set $S=S_N \cup S_X$, such that $|S|=k+\tau$ and $S_X=S\cap X$, we have that every minimizer contains $S_X$. This follows by definition of $X$, since for any two subsets $P,Q$ of $S$ with $|P|=|Q|= k$ and $P$ disjoint with $X$ but $Q\cap X \neq \emptyset$, we have by monotonicity $f(Q)\geq f(x_i)= (k+1)f(a_1)> kf(a_1)$ and by submodularity $kf(a_1) \geq f(P)$. This implies that $X$ is the minimizer of $OPT(k,N,0)\cup X$ and hence $f(OPT(k,N,0))\leq g(OPT(k+\tau,N\cup X,\tau))$ 
			
			For the other direction, consider the set $OPT(k+\tau,N\cup X,\tau)$ and define,
			$$M=OPT(k+\tau,N\cup X,\tau) \cap X$$ 
			Next, observe that carving out an arbitrary set $B$ of size $\tau-|M|$ from $OPT(k+\tau,N\cup X,\tau)-M$ will give us the set $$C=OPT(k+\tau,N\cup X,\tau)-M-B$$ 
			of size $k+\tau-(|M| +\tau-|M|)=k$. Also note that by design, $C\subseteq N$ and hence $f(C)\leq f(OPT(k,N,0))$, but by definition, we have that $g(OPT(k+\tau,N\cup X,\tau))\leq f(C)$. This gives us the other direction and we have $g(OPT(k+\tau,N\cup X,\tau))=f(OPT(k,N,0))$. 
			
			To complete the reduction we need to show how to obtain an $\alpha$-approximate solution to $(k,N,0)$ given an $\alpha$-approximate solution to $(k+\tau,N\cup X,\tau)$. Let $S=S_N \cup S_X$ be such a solution i.e.\ a set of size $k+\tau$ with $S_X=S\cap X$, such that $g(S)\geq \alpha g(OPT(k+\tau,N\cup X,\tau))$. Now consider an arbitrary subset $S'_N$ of $S_N$ of size $\tau-|S_X|$. Observe that $|S_N-S'_N|=|S|-|S_X|-(\tau-|S_X|)=k$ and further $f(S_N-S'_N)\geq g(S)\geq \alpha g(OPT(k+\tau,N\cup X,\tau))=\alpha f(OPT(k,N,0))$, by definition. Hence the set $S_N-S'_N\subseteq N$ is an $\alpha$-approximate solution to $(k,N,0)$ that, given $S$, can be obtained in polynomial time/queries. \qed
			
		\end{proof}
	 \subsection{Tight analysis of Algorithm \ref{a0.387}}	\label{app2} 
	     \begin{thm}\label{0.387}
	   The 0.387-algorithm is $\frac{1}{2}\beta(0.5,\frac{k-2}{k-1})(> 0.387$ asymptotically) approximate.
	   \end{thm}
	   \begin{proof}
	     Let $OPT=g(OPT(k,N,1))$, $A$ be the output of the 0.387-algorithm and $a'_1$ be the first element added to $A$ apart from $a_1$. The case $z=a_1$ is straightforward since $f(A-a_1)\geq \beta(0,1)f(OPT(k-1,N-a_1,0))\geq \beta(0,1)OPT$ where the last inequality follows from Lemma \ref{fng}. So assume $z\neq a_1$. Further, let $f(z|A-a_1-z)=\eta f(A-a_1)$ which implies that $f(a'_1)\geq f(z)\geq f(z|A-a_1-z)= \eta f(A-a_1)$ and now from Lemma \ref{key} with $N$ replaced by $N-a_1$, $A$ replaced by $A-a_1$ and thus $k$ replaced by $k-1$, $S=a'_1$ with $s=1$ and $l=k-1-|S|=k-2$, we get,
	     $$f(A-a_1)\geq \beta(\eta,\frac{k-2}{k-1}) f(OPT(k-1,N-a_1,0))$$
	     This together with Lemma \ref{fng} implies, $f(A-a_1)\geq \beta(\eta,\frac{k-2}{k-1})OPT$.
	     Also, we have by definition,
	     $$f(A-a_1-z)=(1-\eta)f(A-a_1)\geq (1-\eta)\beta(\eta,\frac{k-2}{k-1})OPT$$
	    Further, we have,
	     \begin{eqnarray*}
	     	g(A) &&\geq \max\{ f(a_1), f(A-a_1-z)\}\\
	 &&\geq \max\{ f(z|A-a_1-z), f(A-a_1-z)\}\\
	 && \geq  \max\{\eta \beta(\eta,\frac{k-2}{k-1}), (1-\eta)\beta(\eta,\frac{k-2}{k-1}) \} OPT \\
	 && \geq 0.5\beta(0.5,\frac{k-2}{k-1}) OPT\quad [\text{for } \eta=0.5]\\
	&&\xrightarrow{k\to \infty}   0.387 OPT 
	\end{eqnarray*}
	\end{proof}
	We now give an instance where the above analysis is tight. Let the algorithm start with a maximum value element $a_1$, then pick $a_2$, and then add the set $C$, such that the output of the algorithm is $a_1\cup a_2 \cup C$, with $C$ being a set of size $k-2$. Let $f(a_1)=1, f(a_2)=1, f(C)=1$ with $f(a_1+C)=1, f(a_1+a_2)=2, f(a_2+C)=2$ i.e.\ $C$ copies $a_1$. Hence $f(a_1+a_2+C)=2$ and $g(a_1+a_2+C)=f(a_1+C)=1$. 
	
	Let $OPT(k,N,1)$ include $a_2$, a copy $a'_2$ of $a_2$ (so $f(a'_2)=1, f(a_2+a'_2)=1$) and a set $D$ of $k-2$ elements of value $\frac{1}{(k-2)\beta(0,1)}$ each, such that $f(OPT(k,N,1))=1+(k-2)\frac{1}{(k-2)\beta(0,1)}=1+\frac{e}{e-1}=\frac{2}{\beta(0.5,1)}$. Observe that the small value elements are all minimizers and $g(OPT(k,N,1))\approx \frac{2}{\beta(0.5,1)}$ as $k$ becomes large. Note that $f(D)=\frac{f(C)}{\beta(0,1)} $ and we can have sets $C$ and $D$ as above based on the worst case example for the greedy algorithm given in \cite{nem}.This 
	proves that the inequality in Lemma \ref{key} is tight.
	
		  \subsection{Analysis of Algorithm \ref{a0.55}}\label{app3} 
		  \begin{thm}\label{0.55}
		  	Algorithm \ref{a0.55} is $0.5547-\Omega(1/k)$ approximate.
		  \end{thm}
		  \begin{proof} 
		  	Let $A$ denote the output and $A_0\subset A$ denote $\{a_1,a_2\}$. Due to submodularity, there exists at most two distinct $x\in A$ with $f(x|A-x)>\frac{f(A)}{3}$. Additionally, for every $x\not\in A_0$, we have that $f(x|S)\leq f(a_1)$ and $f(x|S)\leq f(a_2|a_1)$ for arbitrary subset $S$ of $A$ containing $A_0$ and $x\not\in S$. This implies that that $2f(x|S)\leq f(A_0)\leq f(S)$, which gives us that $f(x|S)\leq \frac{f(S+x)}{3}$.
		  	
		  	Note that due to condition in Phase 1, the algorithm ignores $a_1$ even if it is not a minimizer, as long as its marginal is more than a third the value of the set at that iteration. At the end of Phase 1, if $a_2$ has marginal more than third of the set value, then it is ignored until its contribution/marginal decreases. Phase 3 adds greedily (without ignoring any element added). As argued above, no element other than $a_1,a_2$ can have marginal more than a third of the set value at any iteration, so during Phase 2 we have that $a_2$ is also a minimizer. 
		  	
		  	We will now proceed by splitting into cases. Denote $g(OPT(k,N,1))$ as $OPT$ and recall from Lemma \ref{fng}, $OPT\leq f(OPT(k-1,N-a_i,0))$ for $i\in\{1,2\}$. Also, let the set of elements added to $A_0$ during Phase 1 be $U=\{u_1,\dots,u_p\}$, similarly elements added during Phases 2 and 3 be $V=\{v_1,\dots,v_q\},$$ W=\{w_1,\dots,w_r\}$ respectively, with indexing in order of addition to the set. Finally, let $\alpha_p=\frac{p-2}{k-1}\, ,\, \alpha_q=\frac{q-1}{k-1}\, ,\, \alpha_r=\frac{r}{k-1}\, ,\, \alpha=\alpha_p+\alpha_q+\alpha_r=\frac{k-5}{k-1}$, and assume $k\geq 8$.
		  	\\
		  	
		  	\textbf{Case 1:} Phase 2,3 do not occur i.e.\ $p=k-2, q=r=0$.
		
		   Since we have,
		   \begin{eqnarray}
		   	f(A-a_1)&&\overset{(a)}{\geq} f(a_2) + \beta\Big(0,\frac{k-2}{k-1}\Big)(f(OPT(k-1,N-a_1,0))-f(a_2))\notag\\
		   	&&\geq \beta\Big(0,\frac{k-2}{k-1}\Big)f(OPT(k-1,N-a_1,0))\notag\\
		   	&& \overset{(b)}{\geq} \beta\Big(0,\frac{k-2}{k-1}\Big)OPT \label{a-a1},
		   \end{eqnarray} 
		   where $(a)$ follows from Lemma \ref{nem} and $(b)$ from Lemma \ref{fng}.  
		   This deals with the case $z=a_1$. If $z=u_p$, we have, 
		     \begin{eqnarray*}
		     	f(A-u_p)\geq f(A-u_p-a_1)&&\overset{(c)}{\geq} \beta\Big(0,\frac{k-3}{k-1}\Big) f(OPT(k-1,N-a_1,0))\\
		     	&&\geq \beta \Big(0,\frac{k-3}{k-1}\Big)OPT,
		     \end{eqnarray*}
		  		where $(c)$ is like (\ref{a-a1}) above but with $k-2$ replaced by $k-3$ when using Lemma \ref{nem}. 
		  	 Finally, let $z\not\in \{a_1,u_p\}$, then due to the Phase 1 termination criteria, we have $f(a_1|A-a_1-u_p)\geq f(A-u_p)/3$,  which implies that,
		  		\begin{eqnarray}
		  		2 f(a_1|A-a_1-u_p)&&\geq f(A-a_1-u_p)\label{add}
		  		\end{eqnarray}
		  		Now letting $\eta=\frac{f(z|A-a_1-u_p-z)}{f(A-a_1-u_p)}$, we have by submodularity $f(z|A-z)\leq \eta f(A-a_1-u_p)$ and by definition $f(A-a_1-u_p-z)=(1-\eta)f(A-a_1-u_p)$. Using the above we get, 
		  		\begin{eqnarray}
		  		&&f(A-z)\overset{(d)}{\geq} \max\{f(a_1), f(A-u_p-z)\}\notag\\
		  		&&\geq \max\{f(z|A-a_1-u_p-z), f(A-a_1-u_p-z)+ f(a_1|A-a_1-u_p-z)\}\notag\\
		  		&&\geq \max\{\eta f(A-a_1-u_p), (1-\eta)f(A-a_1-u_p)+ f(a_1|A-a_1-u_p)\}\notag\\
		  		&&\geq \max\{\eta, (1-\eta)+ \frac{1}{2}\} f(A-a_1-u_p) \label{semi}
		  		\end{eqnarray}
		  		where $(d)$ follows from monotonicity and the fact that $a_1\in A-z$ and $A-u_p-z\subset A-z$.
		  		Now, from Lemma \ref{key} with $S=\{a_2,u_1\}$, $l=p-2$, $k$ replaced by $k-1$, $N$ by $N-a_1$ and $s=1$, we have, $f(A-a_1-u_p)\geq \beta(\eta,\alpha_p)f(OPT(k-1,N-a_1,0))\geq \beta(\eta,\alpha_p)OPT$. Substituting this in (\ref{semi}) above we get,
		  		\begin{eqnarray}
		  		f(A-z)&&\max\{\eta, \frac{3}{2}-\eta\}\beta(\eta,\alpha_p) OPT  \notag\\
		  		 		&&\geq \frac{3}{4}\beta\Big(\frac{3}{4},\alpha_p\Big) OPT\quad [\eta=3/4]\notag\\
		  				&&> \beta(0,\alpha_p) OPT=\beta\Big(0,\frac{k-4}{k-1}\Big) OPT\label{case1}
		  	\end{eqnarray}
		  	\textbf{Case 2:} Phase 2 occurs, 3 doesn't i.e.\ $p+q=k-2$ and $q>0$.
		
		  	As stated earlier, during Phase 2, $a_2$ is the minimizer of $A_0\cup U\cup (V-v_q)$.
		  	We have $g(A)\geq g(A-v_q)= f(A-v_q-a_2)=f(a_1+U) + f(V-v_q|a_1+U)$. Further, since the addition rule in Phase 2 ignores $a_2$, we have from Lemma \ref{nem}, $f(V-v_q|a_1+U)\geq \beta(0,\alpha_q)(OPT-f(a_1+U))$, and $f(a_1+U)\geq \beta(0,\alpha_p)OPT$ follows from the previous case (to see this, suppose that the algorithm was terminated after Phase 1 and note that $z=a_2$ falls under the scenario $z\not\in\{a_1,u_p\}$). Using this,
		  	\begin{eqnarray}
		  		g(A-v_q)&&\geq f(a_1+U)+\beta(0,\alpha_q)(OPT-f(a_1+U))\notag\\
		  		\frac{f(A-z)}{OPT}&&\geq (1-\beta(0,\alpha_q))\beta(0,\alpha_p) +\beta(0,\alpha_q)\notag\\
		  		&&= \beta(0,\alpha)= \beta\Big(0,\frac{k-5}{k-1}\Big)\label{case2}
		  	\end{eqnarray}
		  	\textbf{Case 3:} Phase 3 occurs i.e., $r>0$.
		  	
		  	We consider two sub-cases, $z\in A-W$ and $z\in W$. Suppose $z\in A-W$. Due to $f(z|A-W-z)\leq f(A-W)/3$ we have, $f(A-W)\leq\frac{3}{2}f(A-W-z)$.
		  		Also, $f(W|A-W-z)\geq f(W|A-W)\geq\beta(0,\alpha_r)(OPT-f(A-W))$. Then using this along with the previous cases,
		  		\begin{eqnarray*}
		  			f(A-z)&&= f(A-W-z)+ f(W|A-W-z)\\
		  			&&\geq f(A-W-z)+ \beta(0,\alpha_r)(OPT-f(A-W))\\
		  			&&\geq f(A-W-z)+ \beta(0,\alpha_r)(OPT-\frac{3}{2}f(A-W-z))\\
		  			&&\geq (1-\frac{3}{2}\beta(0,\alpha_r))f(A-W-z) + \beta(0,\alpha_r)OPT\\
		  			&&\geq (1-\frac{3}{2}\beta(0,\alpha_r))\beta(0,\alpha_p+\alpha_q) OPT + \beta(0,\alpha_r)OPT\quad [\text{from (\ref{case1}),(\ref{case2})}]\\
		  			\frac{f(A-z)}{OPT}
		  			&&\geq 0.5-\frac{3}{2e^{\alpha}} + \frac{1}{2} (e^{-(\alpha_p+\alpha_q)} +e^{-\alpha_r})\\
		  			&&\geq 0.5-\frac{3}{2e^\alpha}+e^{-\alpha/2}\quad [\text{for } \alpha_r=\alpha_p+\alpha_q=\alpha/2]\\
		  			&&= 0.5- \frac{3}{2e^\frac{k-4}{k-1}}+e^{-\frac{k-4}{2(k-1)}}\xrightarrow{k\to \infty} 0.5547
		  		\end{eqnarray*}
		  
		  		Now, suppose $z\in W$, then note that for $p+q\geq 6$ we have either $p\geq 3$, and hence due to greedy additions $f(z|A-z)\leq f(\{u_1,u_2,u_3\})/3\leq f(A-W)/3$, or $q\geq 3$, and again due to greedy additions $f(z|A-z)\leq f(a_1+U\cup \{v_1,v_2\})/3\leq f(A-W)/3$. 
		  		
		  		If $q>0$, then note that $f(A-W)\geq f(A-W-v_q)\geq\frac{3}{2} f(A-W-v_q-a_2)$ due to the Phase 2 termination conditions. Now we reduce the analysis to look like the previous sub-case through the following,
		  		\begin{eqnarray*}
		  			f(A-z)&&= f(A-W)+f(W|A-W)-f(z|A-z)\\
		  			&&\geq f(A-W)+ \beta(0,\alpha_r) (OPT-f(A-W)) -f(z|A-z)\\
		  			&&\geq (1-\beta(0,\alpha_r))f(A-W) +\beta(0,\alpha_r)OPT -f(A-W)/3\\
		  			&&\geq \Big(1-\frac{3}{2}\beta(0,\alpha_r)\Big) \frac{2}{3}f(A-W) +\beta(0,\alpha_r)OPT\\
		  			&&\geq \Big(1-\frac{3}{2}\beta(0,\alpha_r)\Big)f(A-W-v_q-a_2) +\beta(0,\alpha_r)OPT\\
		  		\end{eqnarray*}
		  		Which since $f(A-W-v_q-a_2)\geq \beta\Big(0,\frac{p+q-3}{k-1}\Big) OPT$ from (\ref{case2}), leads to the same ratio asymptotically as when $z\in A-W$. The case $q=0$ can be dealt with similarly by using $f(A-W)\geq f(A-W-u_p)\geq\frac{3}{2} f(A-W-u_p-a_1)$
		  		
		  		If $p+q<6$, then let $f(z|A-z)=\eta f(A)$. Now we have,
		  		$f(A-W)\geq f(A_0)=f(a_1)+f(a_2|a_1)\geq 2f(z|A-z)= 2\eta f(A)$,
		  		which further implies that $f(A-W+w_1)\geq 3 \eta f(A)$ since $z\in W$. Then proceeding as in Lemma \ref{key} with $k$ replaced by $k-1$, $S=A-W+w_1$ and hence $s=3$ and finally $l=k-|S|=k-(p+q+2+1)\leq k-8$ gives us, 
		  		$$f(A)\geq\beta(3\eta,\frac{k-8}{k-1})f(OPT(k-1,N,0))\geq \beta(3\eta,\frac{k-8}{k-1}) OPT$$
		  		Then using Lemma \ref{subkey} we have, $f(A-z)\geq(1-\eta)\beta(3\eta,\frac{k-8}{k-1})OPT \geq \beta(0,\frac{k-8}{k-1})OPT$. \qed
		  
		 \end{proof}

\end{document}